\documentclass[prx,longbibliography,twocolumn,showpacs,nofootinbib,superscriptaddress,notitlepage]{revtex4-1}
\usepackage{amsmath}
\usepackage{amssymb,bm}
\usepackage{amsthm}
\usepackage{mathtools}
\usepackage{url}

\newtheorem{innercustomgeneric}{\customgenericname}
\providecommand{\customgenericname}{}
\newcommand{\newcustomtheorem}[2]{%
  \newenvironment{#1}[1]
  {%
   \renewcommand\customgenericname{#2}%
   \renewcommand\theinnercustomgeneric{##1}%
   \innercustomgeneric
  }
  {\endinnercustomgeneric}
}

\newcustomtheorem{customthm}{Theorem}
\newcustomtheorem{definitionBob}{Definition}

\usepackage{color,dsfont}
\usepackage{graphicx}
\usepackage{ragged2e}
\usepackage[colorlinks=true, hyperindex, breaklinks, linkcolor=blue, urlcolor=blue, citecolor=blue]{hyperref} 
\usepackage[normalem]{ulem}
\usepackage[capitalise]{cleveref}
\usepackage{mathrsfs}
\usepackage[caption=false]{subfig}
\usepackage{mathtools}
\usepackage{float}
\usepackage{verbatim}
\usepackage{latexsym}
\usepackage{amsmath}
\usepackage{algorithm}
\usepackage{algorithmicx}
\usepackage{algpseudocode}
\usepackage{amssymb}
\usepackage{setspace}
\usepackage{amsfonts}
\usepackage{stmaryrd}
\usepackage{xcolor}
\usepackage{enumitem}
\usepackage{lipsum}

\newtheorem{corollary}{Corollary} 
 
\newtheorem{proposition}{Proposition}

\begin{document}

\title{Adaptive Deformation of Color Code in Square Lattices with Defects}

\author{Tian-Hao Wei}
\affiliation{Laboratory of Quantum Information, University of Science and Technology of China, Hefei, 230026, China}

\author{Jia-Xuan Zhang}
\email{zhangjx1@iai.ustc.edu.cn}
\affiliation{Institute of Artificial Intelligence, Hefei Comprehensive National Science Center, Hefei, Anhui, 230088, China}

\author{Jia-Ning Li}
\affiliation{Laboratory of Quantum Information, University of Science and Technology of China, Hefei, 230026, China}

\author{Wei-Cheng Kong}
\affiliation{Origin Quantum Computing Technology (Hefei) Co., Ltd., Hefei, Anhui, 230088, China}

\author{Yu-Chun Wu}
\email{wuyuchun@ustc.edu.cn}
\affiliation{Laboratory of Quantum Information, University of Science and Technology of China, Hefei, 230026, China}
\affiliation{Anhui Province Key Laboratory of Quantum Network, University of Science and Technology of China, Hefei, 230026, China}
\affiliation{Institute of Artificial Intelligence, Hefei Comprehensive National Science Center, Hefei, Anhui, 230088, China}

\author{Guo-Ping Guo}
\affiliation{Laboratory of Quantum Information, University of Science and Technology of China, Hefei, 230026, China}
\affiliation{Anhui Province Key Laboratory of Quantum Network, University of Science and Technology of China, Hefei, 230026, China}
\affiliation{Institute of Artificial Intelligence, Hefei Comprehensive National Science Center, Hefei, Anhui, 230088, China}
\affiliation{Origin Quantum Computing Technology (Hefei) Co., Ltd., Hefei, Anhui, 230088, China}

\begin{abstract}
	
	Quantum error correction is a crucial technology for fault tolerant quantum computing. On superconducting platforms, hardware defects in large scale quantum processors can disrupt the regular lattice structure of topological codes and impair their error correction capabilities. Although defect adaptive methods for surface codes have been extensively studied, other topological codes such as color codes still lack a systematic framework for handling defects. To address this issue, we propose a universal superstabilizer scheme applicable to data qubit defects in arbitrary stabilizer codes. Based on this scheme, we develop concrete repair methods for isolated defects of both internal data qubits and ancilla qubits in color codes defined on square lattices. Furthermore, for ancilla qubit defects, we present two optimization schemes. One scheme reuses neighboring ancilla qubits, and the other employs iSWAP gates. Unlike conventional approaches that directly disable neighboring data qubits and thus cause resource waste, both of our schemes avoid such waste and consequently achieve a lower logical error rate.Integrating the above techniques, we construct a comprehensive defect adaptive architecture for color codes to handle various defect clusters. We also show that our scheme supports a full transversal Clifford gate set and lattice surgery operations. These results provide a systematic theoretical pathway for deploying robust and low overhead color codes on defective quantum hardware.
	
\end{abstract}
\maketitle
%
%
\section{Introduction}
\label{sec:Intro}
Quantum computing, with its parallel processing capability that transcends classical computing paradigms, demonstrates disruptive application potential in fields such as cryptanalysis, quantum simulation, material design, and optimization algorithms \cite{shor1999polynomial,georgescu2014quantum,guo2024harnessing,grover1996fast}. However, physical qubits possess inherent fragility, rendering them highly susceptible to environmental noise, control errors, and relaxation effects \cite{krantz2019quantum,terhal2015quantum}. Quantum coherence is difficult to maintain over extended periods, severely constraining the reliable execution of large-scale quantum algorithms. Quantum error correction (QEC) serves as the key technology to overcome this bottleneck \cite{gaitan2008quantum,campbell2017roads}. By constructing logical qubits through redundant encoding, QEC effectively suppresses the accumulation of physical noise, forming the foundational support for achieving fault-tolerant quantum computing and advancing quantum systems from principle validation to practical implementation.

Practical quantum algorithms typically require thousands of logical qubits with logical error rates below $10^{-10}$, corresponding to a hardware scale on the order of millions of physical qubits \cite{gidney2021factor}. During the large-scale integration of quantum chips, issues such as fabrication imperfections, frequency mismatches \cite{kjaergaard2020superconducting,bilmes2020resolving}, and interference from two-level systems \cite{klimov2018fluctuations,burnett2014evidence} are unavoidable. Existing studies indicate that approximately 1\%–2\% of qubits in current superconducting quantum chip fabrication processes exhibit defects \cite{lisenfeld2019electric,smith2022scaling}. Such hardware defects can lead to qubit failure, degraded entangling gate fidelity, or even complete device failure. Topological quantum encoding relies heavily on regular and complete lattice structures to achieve stable encoding of logical qubits, and hardware defects directly disrupt the topological structure, reduce the code distance, and cause a sharp increase in the logical error rate. Therefore, enabling quantum error-correcting codes to efficiently adapt to non-ideal chip defects and operate stably on noisy hardware has become a critical challenge that must be addressed in the scalable development of superconducting quantum systems.

Among various candidate quantum codes, the surface code enjoys a mainstream status due to its high error threshold and scalability \cite{fowler2012surface}, making it the first to face the challenge of hardware defects and prompting extensive subsequent research on defect adaptation that has yielded a wealth of accumulated experience and strategies. Early studies on defects were primarily based on the superstabilizer approach, employing a data-qubit deactivation strategy that deactivates qubits in the vicinity of defective regions to convert stabilizer measurements into gauge measurements, thereby forming a pair of X/Z superstabilizers covering the entire defective region and reconstructing fault-tolerant stabilizers and logical operators \cite{auger2017fault,strikis2023quantum,siegel2023adaptive,lin2024codesign,yin2024surf}. Subsequent work further optimized the construction of superstabilizers by replacing the coarse approach of using a single pair of superstabilizers to cover an entire defect cluster with a more refined method that stabilizes the defect cluster through multiple so-called “bandage-like” superstabilizers \cite{wei2025low}, significantly suppressing the reduction in code distance and achieving a lower logical error rate. When dealing with isolated data-qubit defects, the superstabilizer scheme typically incurs only a reduction of 1 in both the X and Z code distances (i.e., a reduction of 1 in the weight of error chains passing through the defective qubit) and a doubling of time overhead (since X and Z gauge operators are generally non-commuting and must be measured separately). However, for ancilla qubit defects, the existing superstabilizer approach can only handle them by disabling all data qubits surrounding the defective ancilla, which often results in a reduction of 2 in both the X and Z code distances.

To address ancilla qubit defects, researchers have successively proposed the Snakes and Ladders scheme \cite{leroux2025snakes} and the LUCI scheme \cite{debroy2025luci} to mitigate the resource overhead caused by the traditionalapproach of directly disabling the data qubits neighboring the defective ancilla. The former measures the stabilizer associated with a defective ancilla qubit by reusing neighboring ancilla qubits, maintaining the code distance in one direction while incurring a reduction of 2 in the other direction, with a time overhead comparable to that of the superstabilizer scheme. The LUCI scheme is based on the concept of mid-cycle surface codes and leverages stabilizer measurements from two consecutive cycles to preserve the original code distance in the presence of ancilla qubit defects, albeit at the cost of doubling the time overhead compared to the superstabilizer scheme. This scheme offers high flexibility, tolerating up to 50\% ancilla qubit defects \cite{debroy2025diamond}, and is currently the only approach capable of handling defects in hex-grid surface codes \cite{higgott2025handling}. In addition, some studies have explored the use of iSWAP gates developed on recent superconducting hardware platforms. By employing iSWAP gates, neighboring ancilla qubits can be rotated around a defective ancilla qubit to extract the syndrome of the defect, thereby preserving the code distance while maintaining a time overhead comparable to that of the superstabilizer scheme \cite{zhou2024halma}, achieving the lowest logical error rate for ancilla qubit defects to date.

Compared with the surface code, the color code \cite{fowler2011two} exhibits superior asymptotic scaling in the number of physical qubits required per logical qubit, i.e., higher encoding rate, and under specific conditions, it has the potential to achieve a lower logical error rate for a given physical error probability and code distance \cite{moussa2016transversal}. More importantly, the color code natively supports transversal fault-tolerant implementation of the entire Clifford gate set \cite{moussa2016transversal,kubica2015universal}, a property that directly eliminates the substantial overhead associated with magic state distillation or complex gate teleportation protocols required for non-transversal gates in the surface code, thereby significantly reducing the resource consumption and circuit complexity for fault-tolerant logical gate execution \cite{lee2025low,beverland2021cost,lacroix2025scaling,sales2025experimental}. These structural advantages position the color code as a promising candidate for medium-to-large-scale fault-tolerant quantum computing architectures. However, in contrast to the relatively mature defect adaptation framework of the surface code, other important topological codes such as the color code still lack systematic and comprehensive solutions in this regard.

To this end, this paper conducts a series of studies on defect issues in the 6.6.6 color code with triangular boundaries on a square lattice, which represents the current mainstream quantum chip architecture \cite{gao2025establishing,google2025quantum}. First, we propose superstabilizer schemes for isolated data-qubit defects and ancilla qubit defects, respectively. Compared with the baseline defect-free color code circuit, these schemes incur no additional time overhead while introducing a code distance reduction of 1 for data-qubit defects and 2 for ancilla qubit defects. To address the resource overhead introduced by ancilla qubit defects, we further propose two optimization schemes: the Neighbor-Assisted scheme based on neighboring ancilla reuse, and the iSWAP-Mediated scheme based on iSWAP gates. The former, at the cost of twice the time resources, suppresses the code distance reduction from 4 to 2 under the considered defect cluster distribution; the latter completely eliminates the code distance reduction caused by isolated ancilla qubit defects at the cost of 1.5 times the time resources. Both schemes effectively improve the code distance and qubit utilization, significantly suppressing the accumulation of logical errors. In addition, this paper presents and proves a universal superstabilizer scheme for data-qubit defects in arbitrary stabilizer codes. Based on this scheme, we propose an adaptive adapter architecture that achieves efficient deployment and stable operation of color codes on non-ideal lattices through topological deformation of defective lattices and precise identification of superstabilizers. This provides critical support for the scalable development of quantum systems. Finally, we discuss and provide solutions for logical operations in the presence of defects, including transversal Clifford gates and lattice surgery, thereby establishing a reliable technological pathway toward large-scale universal fault-tolerant quantum computing.

The structure of this paper is as follows. First, in Section II, we introduce the fundamentals of color codes and the benchmark color code circuit adopted in this work. Next, in Section III, we discuss the handling of isolated data-qubit defects, ancilla qubit defects, and coupler defects within the color code, and present the corresponding numerical simulation results. In Section IV, we provide a proof of the universal superstabilizer scheme, and based on this, we elucidate the boundary and corner characteristics of color code defects as well as the complete procedure of the defect adapter, together with numerical simulation results for defect clusters. Then, in Section V, we discuss logical operations on defective color-code lattices. Finally, in Section VI, we conclude the paper and discuss potential directions for future research.

\section{Backround}
\label{sec:Back}
To detect and correct the impact of physical errors on logical information, quantum error-correcting codes encode logical qubits into multiple physical qubits and repeatedly measure stabilizers to extract syndromes, thereby inferring error types. This section focuses on the definition of the 6.6.6 triangular color code in the defect-free case and its mapping to hardware.

The color code is an important class of topological quantum error-correcting codes, introduced by H. Bombin and M.A. Martin-Delgado in 2006 \cite{bombin2006topological}. This work focuses on an important subclass of color codes, defined on a hexagonal lattice, which is often referred to in the literature as the 6.6.6 triangular color code. This color code is constructed on a trivalent, three-face-colorable graph, where each vertex has degree three and all faces can be properly colored with three colors such that adjacent faces receive different colors \cite{kesselring2024anyon}. The code considered in this work adopts a triangular boundary (as shown in the Fig.~\ref{fig:colorcode}(a)), which results in a triangular color code encoding exactly one logical qubit, thereby facilitating fault-tolerant quantum computation. On such a lattice, each physical qubit is placed at a vertex, and each face (including both triangular and hexagonal faces) is associated with two stabilizer operators: an $X$-type stabilizer and a $Z$-type stabilizer, defined respectively as the tensor product of Pauli $X$ operators and Pauli $Z$ operators acting on all data qubits belonging to that face. All face stabilizers commute with one another and collectively generate the stabilizer group that defines the logical subspace. For a finite lattice containing $n$ physical qubits, the number of logical qubits is given by $k = n - |\mathcal{S}|$, where $|\mathcal{S}|$ is the number of independent stabilizer generators, and the error-correcting capability of the code is determined by the shortest weight of logical operators \cite{zhang2024facilitating} (i.e., the code distance). The measurement of stabilizers is central to the error correction process: periodic measurements yield syndrome information ($+1$ or $-1$), where a $-1$ outcome signals the possible occurrence of errors in the corresponding region, and a decoder then infers the most likely physical error chain.

Mapping the abstract color-code lattice to hardware requires careful consideration of the underlying chip’s connectivity \cite{gidney2023new}. This study is based on a common square-lattice superconducting quantum chip architecture \cite{zhao2022realization,wu2021strong,zhao2022realization,google2023suppressing,arute2019quantum,google2025quantum}, in which each qubit is typically directly coupled only to its four nearest neighbors. To realize the high-weight stabilizer measurements required by the hexagonal faces under such limited connectivity, ancilla qubits and specific circuit structures are necessary. The color-code circuit used in this work is shown in Fig.~\ref{fig:colorcode}(b). In this circuit, each region contains two ancilla qubits, which are first prepared into a Bell pair and then used to sequentially extract the $X$ and $Z$ syndromes. During extraction, the two ancilla qubits simultaneously act as flag qubits, helping to detect errors that may propagate during the measurement. All schemes proposed in this paper are modifications based on this circuit, which serves as the baseline circuit for subsequent discussions.

\begin{figure}[!tbp]
	\centering
	\includegraphics[width=\linewidth]{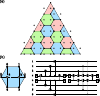}
	\caption{\textbf{Schematic diagram of the 6.6.6 triangular color code:} (a)Color code lattice diagram. Each qubit is coupled only to its four nearest neighbors (up, down, left, right). Solid markers represent data qubits, while hollow markers represent ancilla qubits. (b)Color code circuit diagram. In this circuit, each region contains two ancilla qubits that form a Bell pair, which is then used to extract the X/Z syndromes in sequence. When extracting the X/Z syndrome, the a/b ancilla qubits function as flag qubits. This circuit was proposed by \cite{baireuther2019neural}.}
	\label{fig:colorcode}
\end{figure}

\section{Strategies for Handling Isolated Defects}
\label{sec:Stratey}
This section introduces strategies assuming only an isolated defect exists within the target window. In this section, we focus on defects located in the bulk (i.e., the interior region), but the methods are directly applicable to boundary defects.

\subsection{Data Defect}
\label{subsec:data}
The fundamental starting point of this study lies in canceling the weight on the defective data qubit through products of stabilizers. Consider the local structure surrounding an isolated data qubit defect, where three stabilizer plaquettes are adjacent to the defect.Consider the local structure surrounding an isolated data qubit defect, as shown in Fig.~\ref{fig:data_defect}(a), where three stabilizer plaquettes are adjacent to the defect. Each plaquette contains X, Z, and Y-type stabilizers (where the Y-type is obtained as the product of X and Z). The defective data qubit belongs simultaneously to these three plaquettes; therefore, appropriate products of stabilizers from different plaquettes can be constructed to cancel the weight on the defective qubit. Specifically, all possible combinations of X, Y, and Z stabilizers form a group that can be generated by any two of the three superstabilizers shown in Fig.~\ref{fig:data_defect}(b), where the product of any two superstabilizers determines the third, reflecting the completeness of this generating structure.

Based on the above theoretical framework, the specific operational rules for defect handling are as follows: remove the defective data qubit from its six adjacent stabilizer checks, transforming these six weight-6 stabilizer checks into weight-5 checks. After this modification, adjacent checks of different Pauli types become anti-commuting; consequently, these weight-5 checks no longer yield deterministic measurement outcomes and become gauge checks.

Within this framework, a superstabilizer is defined as the product of multiple gauge checks. The product of any two gauge checks of the same Pauli type commutes with all gauge checks and defect-free stabilizers, yielding deterministic measurement outcomes and forming a superstabilizer of the corresponding Pauli type. In this example, the product of two weight-5 gauge checks of the same Pauli type constitutes a weight-8 superstabilizer.

Through the above operation, the system effectively loses one data qubit and two stabilizers while introducing one logical gauge qubit, whose logical operator can be taken as the product of any two weight-5 gauge checks of different Pauli types. This method effectively “patches” the defective qubit within the code, preserves the number of logical qubits, and introduces no additional time overhead compared to the defect-free color code circuit. However, since the lengths of both X and Z logical chains passing through the defect region are reduced by one, the code distance decreases by one in both directions.

\begin{figure}[!tbp]
	\centering
	\includegraphics[width=\linewidth]{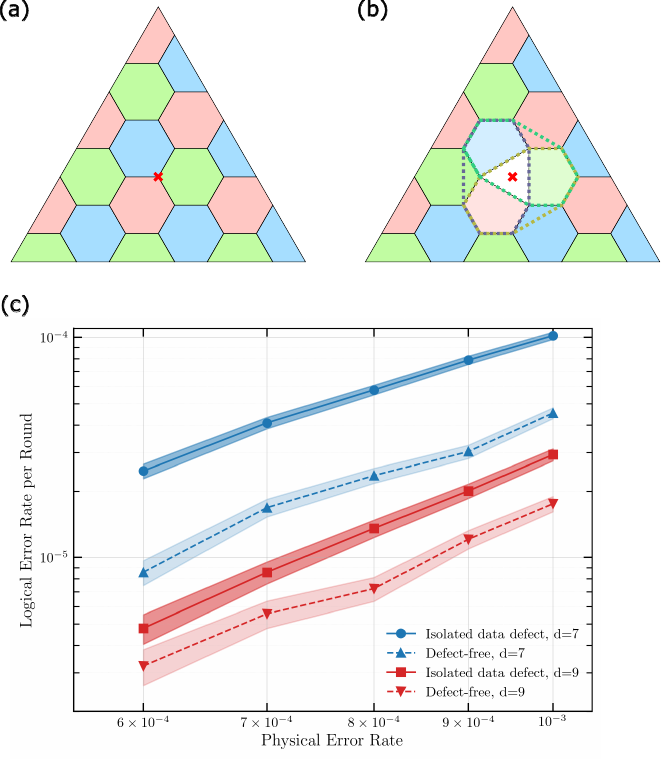}
	\caption{\textbf{Schematic diagrams of an isolated data qubit defect. }(a) A color code patch with code distance \( d = 7 \) containing an isolated data qubit  defect at its geometric center. In all subsequent discussions in this paper, a round refers to one complete cycle of syndrome extraction for all stabilizers.(b) The same color code patch after repairing the isolated data qubit defect using the superstabilizer scheme. (c) Numerical simulation results for color codes with code distances \( d = 7 \) and \( d = 9 \), respectively, obtained under the standard error model. The data qubit defect is placed at the geometric center of the code patch in both cases.}
	\label{fig:data_defect}
\end{figure}

To validate the feasibility of the proposed scheme, we performed numerical simulations using Stim circuits \cite{gidney2021stim} for defect-free configurations and isolated data qubit defect configurations at code distances 7 and 9. More simulation details can be found in Appendix~\ref{app:simulation}. The simulation results are shown in Fig.~\ref{fig:data_defect}(c). At code distance 7, compared with the defect-free configuration, the configuration with an isolated data defect exhibits a slight increase in logical error rate; nevertheless, the trend of logical error rate with respect to physical error rate remains consistent, and no sharp rise in error rate due to the defect is observed. When the code distance is increased to 9, the logical error rates for both configurations decrease by more than one order of magnitude, with the defective configuration still maintaining a relatively low error rate. These results demonstrate that the superstabilizer method proposed in this paper for the color code can effectively suppress logical errors induced by isolated data defects and maintain stable error correction performance as the code distance increases.
\subsection{Ancilla Defect}
\label{subsec:ancilla}
For isolated ancilla qubit defects, the normal measurement of each stabilizer relies on a Bell pair formed by two ancilla qubits; consequently, a defect in any single ancilla qubit renders the corresponding stabilizer unmeasurable. Given this structural characteristic, this paper focuses on the scenario where both ancilla qubits constituting the same stabilizer are defective. Unless otherwise specified, all references to “isolated ancilla qubit defects” in this paper refer to an isolated pair of adjacent ancilla qubit defects. This scenario serves as a representative configuration of ancilla qubit defects that captures the impact of defects on the encoding structure. For this scenario, we propose three feasible handling schemes, which are described and analyzed in detail in this section.
\begin{figure}[!tbp]
	\centering
	\includegraphics[width=\linewidth]{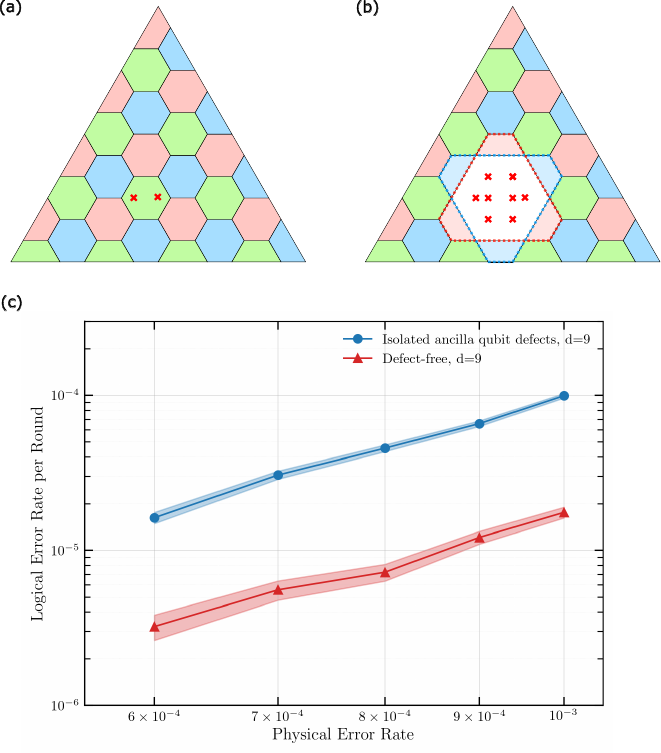}
	\caption{\textbf{Schematic diagrams of isolated ancilla qubit defects.} (a) A color code patch with code distance \( d = 9 \) containing isolated ancilla qubit defects. (b) The same color code patch after repairing the ancilla qubit defects using the superstabilizer scheme. (c) Numerical simulation results for color codes with code distances \( d = 9 \), comparing defect-free configurations with those containing isolated ancilla qubit defects.}
	\label{fig:ancilla_defect_SS}
\end{figure}

\textbf{Superstabilizer:} For ancilla qubit defects, the superstabilizer approach is adapted to handle the failure of a pair of ancilla qubits that perform syndrome extraction. As illustrated in Fig.~\ref{fig:ancilla_defect_SS}(b), the defective ancilla pair and the data qubits directly coupled to them are removed from the stabilizer measurement schedule. This operation transforms the original stabilizer checks into a set of gauge checks with reduced weight. Specifically, by deactivating the six data qubits neighboring the defective ancilla pair, we obtain 12 weight-4 gauge checks.

Within this reconfigured lattice, four independent superstabilizers of weight 12 can be constructed as products of these gauge checks, ensuring that the resulting measurements yield deterministic outcomes. While this approach successfully preserves the logical qubit count and introduces no additional time overhead, it incurs a greater reduction in code distance compared to the data defect case: the code distance decreases by 2 in both bases.

\begin{figure*}[!tbp]
	\centering
	\includegraphics[width=\linewidth]{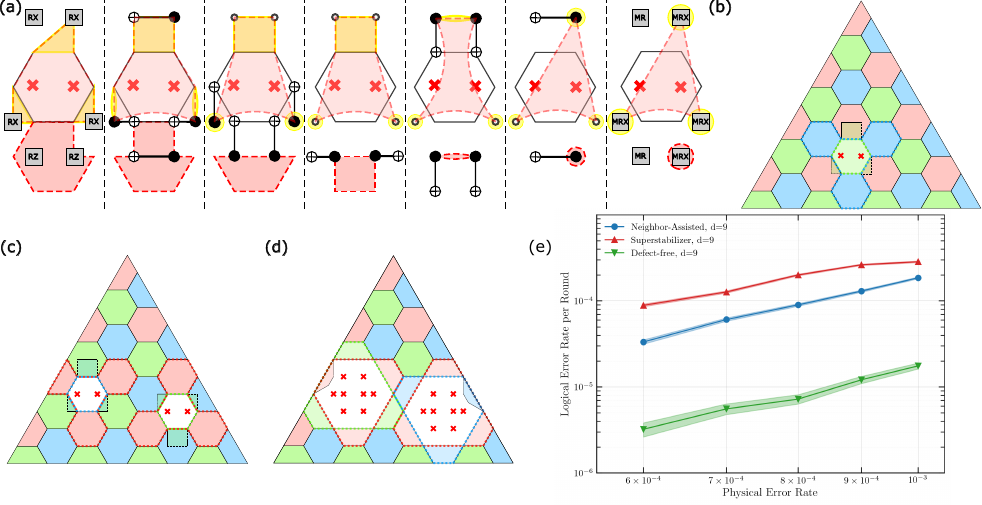}
	\caption{\textbf{Schematic illustrations and numerical results for the Neighbor-Assisted scheme in handling ancilla qubit defects.}  (a) X-type syndrome extraction circuit. The red markers in the figure indicate the X stabilizer flows to be measured, while the yellow markers represent the gauge operators actually measured. While the neighboring ancilla qubits are borrowed to measure gauge operators, the remaining ancilla qubits extract X-stabilizer syndromes normally. (b) Principle of the scheme. The ancilla qubits corresponding to the three red stabilizers surrounding the defective stabilizer measure three gauge operators, whose product forms a superstabilizer that replaces the original defective stabilizer. (c) Example demonstrating the advantage. With the Neighbor-Assisted scheme, the defective color code exhibits a code distance reduction of 2. (d) superstabilizer solution for the same defect distribution as in Fig.~\ref{fig:ancilla_defect_NA}(c), resulting in a code distance reduction of 4. (e) Numerical simulation results comparing the logical error rates of the Neighbor-Assisted scheme and the superstabilizer scheme, corresponding to the circuit configurations in Fig.~\ref{fig:ancilla_defect_NA}(c) and (d), respectively. }
	\label{fig:ancilla_defect_NA}
\end{figure*}
\textbf{Neighbor-Assisted:} Building upon the existing SNL scheme for surface codes \cite{leroux2025snakes}, we propose a defect compensation scheme that utilizes neighboring ancilla qubits, termed the Neighbor-Assisted scheme. As illustrated in Fig.~\ref{fig:ancilla_defect_NA}(b), this scheme utilizes the ancilla qubits associated with the three red stabilizer plaquettes adjacent to the defective stabilizer to measure three gauge operators. The product of these three gauge operators forms a superstabilizer corresponding to the original defective stabilizer. This operation forces the three blue stabilizer plaquettes adjacent to the red region to multiply together due to commutation constraints, collectively forming a larger superstabilizer. By reconstructing the stabilizer measurement paths, we are able to bypass the faulty region.

This scheme resembles the SNL scheme for surface codes \cite{leroux2025snakes} in that both handle defects by reusing neighboring ancilla qubits. However, due to the denser arrangement of stabilizers in color codes, it is not possible to sacrifice the code distance in only one basis (by reducing it by two) while preserving that in the other basis, as achieved by the Neighbor-Assisted scheme. A detailed analysis is provided in the Appendix~\ref{app:Neighbor-Assisted}. Consequently, for color codes, the Neighbor-Assisted scheme must be applied uniformly to both bases, as shown in Fig.~\ref{fig:ancilla_defect_NA}(b), ultimately reducing the code distance by two in both bases.

The implementation of the Neighbor-Assisted scheme consists of two measurement cycles. In the first cycle, all stabilizers except the defective one extract both X-type and Z-type syndromes sequentially in the conventional manner. In the second cycle, the X-type (or Z-type) syndrome of the defective stabilizer is obtained by measuring the corresponding X-type (or Z-type) gauge operators. Specifically, the three X-type (or Z-type) gauge operators are measured sequentially using the ancilla qubits in the red region, and the measurement results are multiplied to indirectly obtain the syndrome value of the defective stabilizer. Meanwhile, the remaining stabilizers not involved in the gauge operator measurements continue to extract syndromes using the standard procedure. Figure~\ref{fig:ancilla_defect_NA}(a) illustrates the timing and control logic of this scheme for the X-type syndrome extraction circuit during the second cycle. The advantage of this scheme is not pronounced for isolated ancilla qubit defects but becomes significant for certain specific defect clusters, such as the one shown in Fig.~\ref{fig:ancilla_defect_NA}(c). For the defect configuration in Fig.~\ref{fig:ancilla_defect_NA}(c), the Neighbor-Assisted scheme reduces the code distance by only two, whereas the conventional superstabilizer scheme (depicted in Fig.~\ref{fig:ancilla_defect_NA}(d)) leads to a code distance reduction of four. The numerical simulation results presented in Fig.~\ref{fig:ancilla_defect_NA}(e) validate this performance difference, demonstrating that the logical error rate of the Neighbor-Assisted scheme is substantially lower than that of the superstabilizer scheme.

\textbf{iSWAP-Mediated:}
For superconducting platforms that support $XX$ and $YY$ interactions, recent advances in gate schemes \cite{chen2024one,chen2025efficient} have enabled the simultaneous implementation of high-fidelity CZ and iSWAP gates. Notably, the CXSWAP gate, which is equivalent to the iSWAP gate, facilitates qubit exchange and entanglement in a single operation, thereby providing a foundation for constructing Bell pairs between ancilla qubits. Inspired by the Halma scheme for surface codes \cite{zhou2024halma}, we leverage this property to propose a third approach for handling isolated ancilla qubit defects, termed the iSWAP-Mediated scheme, as illustrated in Fig.~\ref{fig:circuit_iM}.

This scheme operates in cycles of two full rounds of syndrome extraction, divided into four stages. In the first stage, syndromes are extracted from all intact X-type stabilizers in the standard manner. At the end of the circuit, the CX gates on both sides of the defective region are replaced with CXSWAP gates, which swap the positions of ancillary and data qubits to prepare for Bell pair formation in the subsequent stage, as depicted in steps 1–2 of Fig.~\ref{fig:circuit_iM}. In the second stage, using the two prepared Bell pairs along with iSWAP gates, the scheme simultaneously extracts both X- and Z-type syndromes for the defective stabilizers. It also extracts syndromes for the two Z-type stabilizers in the right neighboring region and the four X-type stabilizers in the remaining neighboring regions, while non-neighboring X-type stabilizers are extracted using the standard approach, as shown in steps 4–15. Because the distribution of ancillary and data qubits changes significantly after the second stage, the third stage performs qubit resetting alongside syndrome extraction. This stage again employs the two Bell pairs and iSWAP gates to simultaneously extract X- and Z-type syndromes for the defective stabilizers, as well as the two X-type stabilizers in the right neighboring region and the four Z-type stabilizers in the other neighboring regions, resetting all qubits except the two side ancilla qubits. The fourth stage extracts syndromes from all intact Z-type stabilizers in the standard manner while resetting the ancilla qubits used in the preceding two stages. Detailed circuits and further details of the complete iSWAP cycle are provided in Appendix~\ref{app:iSWAP-Mediated}.

In terms of time overhead, a defect-free color code circuit requires 20 time steps of two-qubit gate operations and four rounds of measurement/reset operations to complete two full rounds of syndrome extraction. In contrast, the iSWAP-Mediated scheme accomplishes the same two rounds of syndrome extraction in 30 time steps of two-qubit gate operations while still requiring only four rounds of measurement/reset operations. From the perspective of ideal gate-sequence duration, this corresponds to a 50\% increase in time overhead over the defect-free circuit. However, on mainstream superconducting quantum computing platforms, measurement and reset operations typically dominate the cycle time, as they are substantially slower than two-qubit gates. Since the iSWAP-Mediated scheme maintains the same number of measurement/reset operations as the defect-free circuit, the actual increase in time overhead is significantly less than 50\%. Therefore, when measurement and reset times dominate, this scheme effectively handles isolated ancilla qubit defects while preserving the code distance, incurring only modest time overhead and demonstrating favorable hardware adaptability and significant potential for further optimization.

\begin{figure*}[!tbp]
	\centering
	\includegraphics[width=\linewidth]{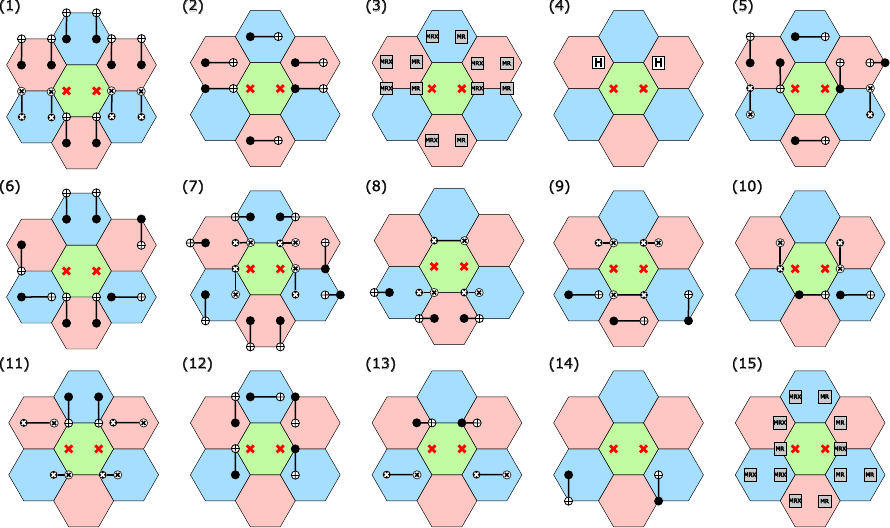}
	\caption{\textbf{Circuit diagram of the iSWAP-Mediated scheme.} The first two time steps correspond to a normal syndrome extraction cycle for all non-defective stabilizers, during which CXSWAP gates are applied in advance to swap ancillary and data qubits, preparing Bell pairs for the subsequent cycle. In the remaining time steps shown, the circuit extracts the X and Z syndromes of the defective stabilizer, the Z syndromes of the upper two regions, and the X syndromes of the other four regions. The complete circuit is provided in Appendix~\ref{app:iSWAP-Mediated}.}
	\label{fig:circuit_iM}
\end{figure*}
\begin{figure}[!tbp]
	\centering
	\includegraphics[width=\linewidth]{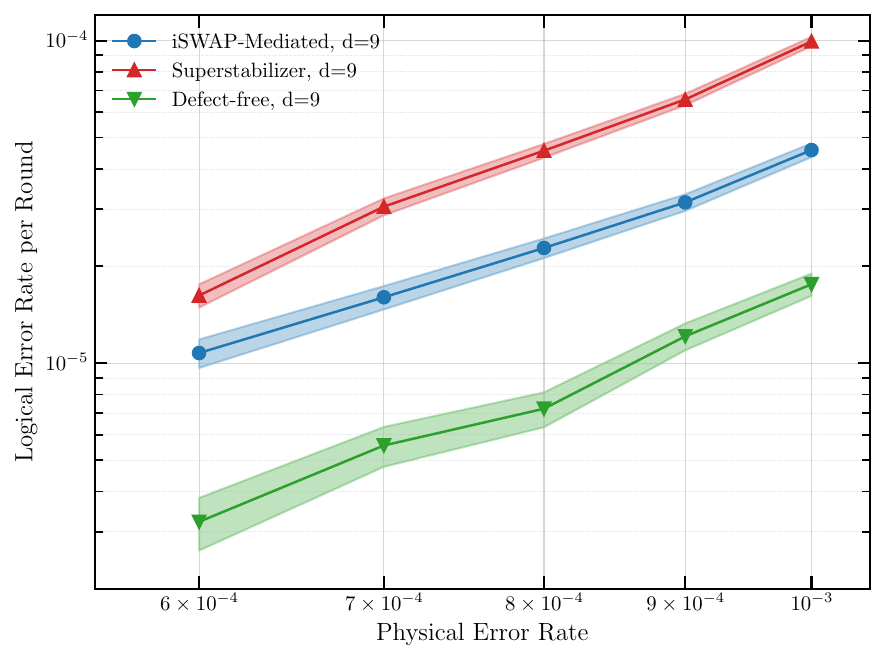}
	\caption{Simulation results for isolated ancilla qubit defects, with the defect distribution as depicted in Fig.~\ref{fig:ancilla_defect_SS}(a).}
	\label{fig:sim_iM}
\end{figure}
\subsection{Coupler Defect}
\label{subsec:coupler}
For an isolated defective two-qubit gate (coupler), the defective coupler can be treated as a data defect and handled using the superstabilizer scheme, which results in a code distance reduction of one along both the X and Z bases. Alternatively, it can be treated as an ancilla defect and addressed through the superstabilizer, Neighbor-Assisted, or iSWAP-Mediated schemes. The first two approaches lead to a code distance reduction of two, whereas the latter incurs no distance loss. A defective coupler connecting two ancilla qubits must be treated as an ancilla defect. Therefore, when iSWAP gates are unavailable, an isolated internal coupler defect is treated as a data defect; when the iSWAP-Mediated scheme can be employed, it is treated as an ancilla defect.

It is worth noting that coupler defects, owing to their higher connectivity compared to data or ancilla defects, may offer optimization opportunities when the Neighbor-Assisted or iSWAP-Mediated schemes are applied. However, no significant advantage has been identified thus far. Consequently, in the following discussion, coupler defects will continue to be classified as either data or ancilla defects as appropriate.
\section{Defective Lattice Color Code Adapter}
\label{sec:Adapter}
In this section, we propose and prove a superstabilizer defect-handling scheme applicable to data qubit defects in arbitrary stabilizer codes. Based on this scheme, we discuss defects at boundaries and corners, and by synthesizing all the above content, we finally propose a defect-adaptive superstabilizer architecture for color codes.
\subsection{Universal superstabilizer Scheme}
\label{subsec:superstabilizer}
By reducing ancilla qubits and coupler defects to data qubit defects, we propose a complete treatment scheme based on the superstabilizer approach for defect clusters formed by multiple adjacent data qubit defects. We demonstrate that after applying this scheme, the code distance decreases at most to the number of defective data qubits. This is illustrated by proving the following propositions.

We first introduce a key tool: for any stabilizer code and any physical qubit, one can always recombine the stabilizer generators so that the \(X\) and \(Z\) operators on that qubit each anticommute with at most one generator. This property lays the foundation for handling defective qubits.

\begin{figure}[!tbp]
	\centering
	\includegraphics[width=\linewidth]{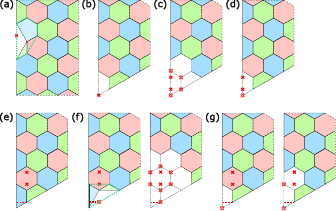}
	\caption{\textbf{Defect handling at boundaries and corners}
		(a) Isolated data qubit defect on the boundary: multiply the two defective stabilizers to form a superstabilizer.  
		(b) Isolated data qubit defect at a corner: directly discard the defective stabilizer.  
		(c) Handling isolated ancilla qubit defects at a corner by disabling neighboring data qubits, reducing the code distance by 2.  
		(d) Handling a corner ancilla qubit defect by leveraging the properties of corner data qubit defects, reducing the ancilla defect to a data qubit defect and lowering the code distance by 1.  
		(e) Schematic of a corner defect cluster consisting of a pair of ancilla qubit defects and a coupler defect.  
		(f) Reducing the coupler defect to an adjacent data qubit defect, reducing the code distance by 3.  
		(g) Reducing the coupler defect to an ancilla qubit defect: after disabling the corner defective stabilizer, the inner ancilla qubit defect becomes a new corner defect, which can be addressed by disabling the stabilizer, resulting in a code distance reduction of 2.}
	\label{fig:defectcluster_boundary}
\end{figure}
\begin{figure*}[!tbp]
	\centering
	\includegraphics[width=0.8\linewidth]{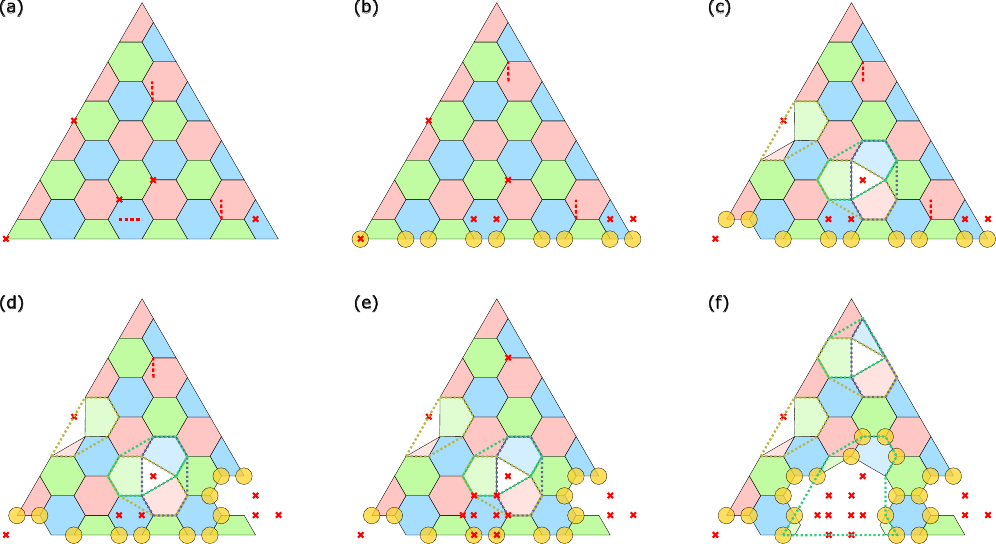}
	\caption{\textbf{Construction Process of the Adaptive Scheme for Color Code Defects}(a) Defect distribution: Initial distribution of data qubit and ancilla qubit defect.
		(b) Ancilla qubit defect handling: Each ancilla qubit defect is treated as if its neighboring ancilla qubit is defective, with the logical operator highlighted in yellow. 	(c) Data qubit defect handling: The stabilizers are initialized as the set in the defect-free case. For each data qubit defect \(d\), all stabilizers that act nontrivially on \(d\) (i.e., with weight one on \(d\)) are identified. One such stabilizer is arbitrarily selected for removal, and the remaining selected stabilizers are updated by multiplying them with the removed stabilizer. (d)Screening and handling of corner ancilla and coupler defects: Coupler defects are temporarily reduced to the ancilla qubit defect they connect. Subsequently, for defects on ancilla qubits that participate in the measurement of only one X and one Z stabilizer, a sequential check is performed within the support set of the X/Z stabilizer to determine whether there exists a data qubit that is stabilized by exactly one X stabilizer and one Z stabilizer, respectively. If such a data qubit exists, the corresponding stabilizer is temporarily removed, and the ancilla/coupler defects in this region are reduced to that data qubit defect, after which the detection process is re-executed.	(e) Reduction of ancilla and coupler defects to data qubits: For the remaining coupler and ancilla qubit defect, coupler defects are reduced to the data qubit defects they connect (with couplers between two ancilla qubits treated as ancilla qubit defects), and ancilla qubit defects are reduced to all their neighboring data qubit defects.  (f) Repetition of the first step: After completing the above reductions, the first step is repeated, ultimately yielding a well-defined set of superstabilizers that compensate for the defect cluster.}
	\label{fig:defectcluster_diag}
\end{figure*}

\begin{proposition}[Stabilizer Normalization]
	For any stabilizer code with stabilizer generators \(\{S_i\}\) and any physical qubit \(q\), there exists a set of stabilizer generators \(\{S'_i\}\) (obtained by invertible \(\mathbb{F}_2\) row operations on the check matrix) such that \(X_q\) and \(Z_q\) each anticommute with at most one generator in \(\{S'_i\}\).
\end{proposition}

\begin{proof}
	Let \(H = (H_x \mid H_z) \in \mathbb{F}_2^{m \times 2n}\) be the stabilizer check matrix, where \(H_x\) (\(H_z\)) indicates the positions of \(X\) (\(Z\)) Pauli operators. Fix qubit \(q\) (w.l.o.g. \(q=1\)) and define the \(m \times 2\) matrix \(M\) whose \(i\)-th row is \(\big( (H_z)_{i,1},\; (H_x)_{i,1} \big)\). The first entry is \(1\) iff \(S_i\) anticommutes with \(X_1\), and the second entry is \(1\) iff \(S_i\) anticommutes with \(Z_1\). Invertible row operations on \(H\) correspond to replacing stabilizer generators by products among themselves; performing the same row operations on \(M\) yields the matrix associated with the new generators. Since \(M\) has only two columns, performing Gaussian elimination (e.g., reducing to row echelon form) yields a matrix \(M'\) where each column contains at most one non‑zero entry; thus at most one row has a \(1\) in the first column (anticommuting with \(X_q\)) and at most one row has a \(1\) in the second column (anticommuting with \(Z_q\)). The corresponding new generators \(\{S'_i\}\) satisfy the required property, and the stabilizer group remains unchanged because row operations preserve the group.
\end{proof}

Using the normalized generators obtained above, we can extend the original stabilizer code into a subsystem code where the defective qubit \(q\) is “isolated” into the gauge group, thereby allowing us to safely discard it. The following proposition describes this expansion process and its effect on the encoding capacity.

\begin{proposition}[Gauge Space Expansion]
	Let \(\{S'_i\}\) be the normalized stabilizer generators obtained from Stabilizer Normalization. Let \(\mathcal A\) be the set of generators that anticommute with \(X_q\) or \(Z_q\) (so \(|\mathcal A|\le 2\)), and let \(\mathcal R = \{S'_i\}\setminus \mathcal A\). Define \(\mathcal S_{\text{new}} = \langle \mathcal R \rangle\). Then:
	\begin{enumerate}
		\item \(X_q, Z_q \in C(\mathcal S_{\text{new}})\), where \(C(\mathcal S_{\text{new}})\) denotes the centralizer of  \(\mathcal S_{\text{new}}\).
		\item Every logical operator \(L\) of the original stabilizer code can be transformed (by multiplying with an element of \(\langle X_q, Z_q\rangle\) and possibly with elements of \(\mathcal A\)) into an operator \(\tilde L\) that commutes with the whole original stabilizer group and acts as identity on qubit \(q\). Consequently, the restriction of \(\tilde L\) to the remaining qubits yields an operator that commutes with \(\mathcal S_{\text{new}}\); this gives a mapping from original logical operators to logical operators of the subsystem code whose stabilizer is \(\mathcal S_{\text{new}}\) and whose gauge group is \(\langle \mathcal A, X_q, Z_q\rangle\). In particular, the new code can encode at least \(k-2\) logical qubits if the original code encoded \(k\) qubits, and the distance bound of Proposition~3 holds.
	\end{enumerate}
\end{proposition}

\begin{proof}
	By construction, every generator in \(\mathcal R\) commutes with \(X_q\) and \(Z_q\), so \(X_q, Z_q \in C(\mathcal S_{\text{new}})\). This proves (1).
	
	For (2), let \(L\) be any logical operator of the original stabilizer code, i.e., \(L\) commutes with all stabilizers and \(L\notin \mathcal S\) (the original stabilizer group). After normalization, there is at most one generator anticommuting with \(X_q\) and at most one anticommuting with \(Z_q\). Let \(A_X\) (if it exists) be the generator that anticommutes with \(X_q\), and let \(A_Z\) (if it exists) be the generator that anticommutes with \(Z_q\). Since each column of the reduced matrix \(M'\) has at most one non‑zero entry, no generator anticommutes with both \(X_q\) and \(Z_q\); hence \(A_X\) and \(A_Z\) are distinct (if both exist). We can choose exponents \(a,b\in\{0,1\}\) such that \(\tilde L = L X_q^a Z_q^b\) commutes with every element of \(\mathcal A\): if \(A_X\) exists, pick \(a\) so that \(\tilde L\) commutes with \(A_X\) (i.e., \(a=1\) when \(L\) anticommutes with \(A_X\), else \(a=0\)); if \(A_Z\) exists, pick \(b\) similarly. Thus \(\tilde L\) commutes with all generators in \(\mathcal A\) and also with \(\mathcal R\) (since \(L\) does and \(X_q,Z_q\) commute with \(\mathcal R\)). Therefore \(\tilde L\) commutes with the whole original stabilizer group \(\langle \mathcal R \cup \mathcal A\rangle = \mathcal S\).
	
	Now \(\tilde L\) may still have non‑trivial support on qubit \(q\). Because \(X_q\) and \(Z_q\) generate all Pauli operators on \(q\) (up to phases) and both commute with \(\mathcal S_{\text{new}}\), we can multiply \(\tilde L\) by an appropriate product \(X_q^u Z_q^v\) to obtain \(L' = \tilde L X_q^u Z_q^v\) such that \(L'\) acts as \(I_q\) on qubit \(q\) (i.e., \(L' = L_{\text{rest}} \otimes I_q\)). Indeed, if \(\tilde L = O \otimes P_q\) with \(P_q \in \{I,X,Y,Z\}\), choose \(u,v\) so that \(P_q X_q^u Z_q^v = I\) (mod phase). Then \(L'\) still commutes with \(\mathcal S\) because \(X_q,Z_q \in C(\mathcal S)\).
	
	The restriction \(L_{\text{rest}}\) of \(L'\) to the remaining qubits commutes with \(\mathcal S_{\text{new}}\) (since \(L'\) commutes with \(\mathcal R\) and \(\mathcal S_{\text{new}}\) is generated by restrictions of operators in \(\mathcal R\), each of which acts as \(I_q\)). If \(L_{\text{rest}}\) belonged to the gauge group \(\langle \mathcal A, X_q, Z_q\rangle\) (restricted to the remaining qubits), then \(L'\) would be in \(\langle \mathcal A, X_q, Z_q\rangle\) and consequently \(L\) would be in the original gauge group \(\langle \mathcal S, X_q, Z_q\rangle\); but \(L\) was a non‑trivial logical operator, so this does not happen. Hence \(L_{\text{rest}}\) is a non‑trivial logical operator of the subsystem code with stabilizer \(\mathcal S_{\text{new}}\) and gauge group \(\langle \mathcal A, X_q, Z_q\rangle\). This construction gives an injection from the original logical quotient \(C(\mathcal S)/\mathcal S\) into the new logical quotient modulo the gauge group, with a kernel of dimension at most \(2\) (corresponding to logical operators that are equivalent modulo \(\mathcal S\) to elements of \(\langle X_q, Z_q\rangle\)). Therefore the new code can encode at least \(k-2\) logical qubits. The distance bound is treated separately in Proposition~3.
\end{proof}

The above expansion allows us to discard one defective qubit, but we need to know how the code distance changes. The next proposition gives a precise lower bound: discarding one qubit reduces the code distance by at most 1.

\begin{proposition}[Distance Lower Bound for Data Defect]
	For any stabilizer code with logical distance \(d\), after applying Gauge Space Expansion and discarding physical qubit \(q\), let \(\mathcal{G}_{\text{new}} = \langle \mathcal{R},\; \mathcal{A},\; X_q, Z_q \rangle\) be the gauge group of the resulting subsystem code. Its logical distance \(d'\) satisfies \(d' \ge d-1\).
\end{proposition}

\begin{proof}
	Let \(M\) be a non‑trivial logical operator of the new code (after discarding \(q\)), i.e., \(M\) commutes with \(\mathcal{S}_{\text{new}}\) but \(M \notin \mathcal{G}_{\text{new}}\). Define \(P = M \otimes I_q\). Then \(P\) commutes with \(\mathcal{S}_{\text{new}}\) but may anticommute with the removed generators in \(\mathcal A\). Denote by \(A_X\) the generator that anticommutes with \(X_q\) (if it exists) and by \(A_Z\) the generator that anticommutes with \(Z_q\) (if it exists). Because \(X_q\) anticommutes with at most one of them and \(Z_q\) with at most one, there exist \(a,b \in \{0,1\}\) such that \(P'' = P X_q^a Z_q^b\) commutes with every generator in \(\mathcal A\). Hence \(P''\) commutes with all original stabilizers, i.e., \(P'' \in C(\mathcal{S})\) of the original stabilizer group.
	
	If \(P'' \in \mathcal{S}\), then \(P''\) can be written as a product of generators from \(\mathcal{R}\) and \(\mathcal A\). Consequently \(M = P'' X_q^a Z_q^b\) lies in the group generated by \(\mathcal{R}\), \(\mathcal A\), and \(X_q,Z_q\), which is exactly \(\mathcal{G}_{\text{new}}\), a contradiction. Therefore \(P''\) is a non‑trivial logical operator of the original stabilizer code, so its weight \(w(P'') \ge d\). Since \(a,b\in\{0,1\}\), the operator \(X_q^a Z_q^b\) has weight at most \(1\). Thus \(w(P'') \le w(M) + 1\), and we obtain \(w(M) \ge d-1\). Taking the minimum over all non‑trivial logical operators yields \(d' \ge d-1\).
\end{proof}

The above proposition shows that we can discard defective qubits one by one, losing at most 1 in distance at each step. Naturally, applying this result repeatedly handles multiple defects, and the measurement process remains feasible in practice. We summarize these conclusions in the following corollary.

\begin{corollary}[Defect Cluster and Measurement Feasibility]
	For any stabilizer code, suppose we sequentially apply Stabilizer Normalization and Gauge Space Expansion to \(f\) physical qubits (defects) and then discard those qubits. Then:
	\begin{enumerate}
		\item \textbf{Distance bound:} The logical distance \(d_f\) of the final subsystem code satisfies \(d_f \ge \max(0,\, d - f)\), where \(d\) is the original distance.
		\item \textbf{Measurement feasibility:} At each step, the new stabilizer group \(\mathcal{S}_{\text{new}}\) can be generated by products of gauge operators of the original code (with the discarded qubits removed). Consequently, measuring the new stabilizers can be reduced to measuring those gauge operators, which is feasible.
	\end{enumerate}
\end{corollary}

\begin{proof}
	(1) We prove by induction on \(f\) that for any nontrivial logical operator \(M_f\) of the final subsystem code (after discarding \(f\) qubits), there exists a logical operator \(L\) of the original stabilizer code with \(w(L) \le w(M_f) + f\). The base case \(f=0\) is trivial. For the inductive step, assume the claim holds for \(f-1\) defects. Consider the code obtained after discarding the first \(f-1\) defects; by the induction hypothesis, any nontrivial logical operator \(M_{f-1}\) of that code satisfies \(w(L) \le w(M_{f-1}) + (f-1)\) for some original logical operator \(L\). Now discard one more qubit \(q_f\). Applying Proposition~3 to \(M_{f-1}\) (viewed as a logical operator of the code before discarding \(q_f\)) yields \(w(M_{f-1}) \ge d_{f-1} - 1\), where \(d_{f-1}\) is the distance of the code before discarding \(q_f\). By the induction hypothesis, \(d_{f-1} \ge d - (f-1)\). Combining these inequalities gives \(w(M_{f-1}) \ge d - f\), which is the desired bound for \(f\) defects. Taking the minimum over all \(M_f\) gives \(d_f \ge d - f\).
	
	(2) Consider a single qubit \(q\). In the original code, every stabilizer generator is itself a gauge operator. After Proposition~1 and Proposition~2, the new stabilizer generators are taken from \(\mathcal{R}\), i.e., those transformed generators \(S'_i\) that commute with \(X_q\) and \(Z_q\). For a Pauli operator on qubit \(q\), commuting with both \(X_q\) and \(Z_q\) forces it to be the identity \(I_q\) (since \(Y_q\) anticommutes with both). Hence each \(S'_i\) has trivial support on \(q\); discarding \(q\) simply removes that factor, yielding an operator \(R_i\) on the remaining qubits. Because the \(S'_i\) originally commuted with all stabilizers, the \(R_i\) commute with each other and generate \(\mathcal{S}_{\text{new}}\). Each \(R_i\) is the restriction of a gauge operator \(S'_i\) of the original code, so it belongs to the gauge group of the original code (restricted to the remaining qubits). Therefore, to measure a new stabilizer \(R_i\), one can measure the original gauge operator \(S'_i\) (a product of original stabilizer generators) and ignore the outcome on the discarded qubits; the measurement outcome directly yields the eigenvalue of \(R_i\). For multiple qubits, we iterate this argument: after discarding one qubit, the remaining code still has the structure of a stabilizer code (with a gauge group), and the same reasoning applies to the next qubit. Hence a feasible measurement scheme exists throughout the whole process.
\end{proof}

In summary, we have proposed and proved a scheme for code restructuring applicable to arbitrary subsystem codes. The scheme handles physical defects by removing defective data qubits and reassigning stabilizer generators. Moreover, we show that the distance of the new code obtained after restructuring decreases by at most \(f\), where \(f\) is the number of defects.
\subsection{Boundary and Corner}
\label{subsec:Boundary}
In the surface code, defects near the boundary are typically handled using boundary deformation. However, the superstabilizer scheme proposed in the previous section does not treat boundary defects as a special case. In fact, from the perspective of stabilizers, whether a defect is located on the boundary (or at a corner) only determines the number of stabilizers it affects.

Taking data qubits as an example, Fig.~\ref{fig:defectcluster_boundary}(a) and (b) illustrate isolated data qubit defects on the boundary and at a corner, respectively. Unlike an interior defect, which affects three pairs of $X$/$Z$ stabilizers, a data qubit defect on the boundary disrupts only two pairs of $X$/$Z$ stabilizers. According to the superstabilizer scheme, one pair of $X$/$Z$ stabilizers can be discarded, and the remaining pair can be multiplied to form a superstabilizer, as shown in Fig.~\ref{fig:defectcluster_boundary}(a). For a data qubit defect at a corner, since it affects only one pair of $X$/$Z$ stabilizers, the affected stabilizers are directly discarded following the scheme, as depicted in Fig.~\ref{fig:defectcluster_boundary}(b).

Both coupler defects and ancilla qubit defects can be reduced to data qubit defects. As previously discussed, an isolated coupler defect can be equivalent to a data qubit defect or an ancilla qubit defect to which it is connected, while an isolated ancilla qubit defect can be equivalent to defects on all its neighboring data qubits, which typically reduces the code distance by $2$, as shown in Fig.~\ref{fig:defectcluster_boundary}(c). However, a more favorable approach exists for handling ancilla qubit defects located at corners, as illustrated in Fig.~\ref{fig:defectcluster_boundary}(d). According to the superstabilizer scheme, if a corner data qubit is defective, since it is stabilized by only a single $X$/$Z$ stabilizer, the affected stabilizer is simply discarded without generating new superstabilizers via multiplication. Once this stabilizer is discarded, circuit measurements naturally no longer involve the defective ancilla qubit. Therefore, a corner ancilla qubit defect can be reduced to a single corner data qubit defect without the need to disable all its neighboring data qubits.

For coupler defects, if the superstabilizer scheme is applied alone, they are typically reduced to the data qubit defects they connect. However, for coupler defects at corners (which can be defined as couplers connected to corner ancilla qubits), an alternative reduction to corner data qubit defects is also possible. Both approaches result in a code distance reduction of at most $1$, though evaluating their precise impact on the code distance is nontrivial. In our proposed defect adapter, we adopt the latter approach for corner coupler defects. This choice is motivated by the fact that disabling corner data qubits and stabilizers can further generate new corner data qubits and stabilizers, potentially converting original interior ancilla qubit defects into corner ancilla qubit defects, thereby reducing the total number of disabled data qubits, as illustrated in Fig.~\ref{fig:defectcluster_boundary}(f)(g).

The term “corner” in the above discussion is not limited to the geometric corners of the code block but is more broadly defined by the stabilizer structure. Specifically, if an ancilla qubit participates in the measurement of only one pair of $X$ and $Z$ stabilizers, and within the support set of these $X$/$Z$ stabilizers there exists a data qubit that is stabilized by exactly one $X$ stabilizer and one $Z$ stabilizer, respectively, then this data qubit is referred to as a corner data qubit, the corresponding $X$/$Z$ stabilizers are called corner stabilizers, and the ancilla qubit that participates only in the measurement of these stabilizers is called a corner ancilla qubit. A more detailed discussion on boundaries is provided in Appendix~\ref{app:Boundary}.
\subsection{Superstabilizer Construction for Defect Clusters}
\label{subsec:Clusters}
We now provide a comprehensive description of the proposed adaptive scheme for color code defects. The figure illustrates the step-by-step construction of the superstabilizer architecture using a defect grid of the color code with a code distance of 9 as an example. According to the propositions we have proved, all data qubit defects can be handled as shown in Fig.~\ref{fig:defectcluster_diag}(a).
\begin{figure}[!tbp]
	\centering
	\includegraphics[width=\linewidth]{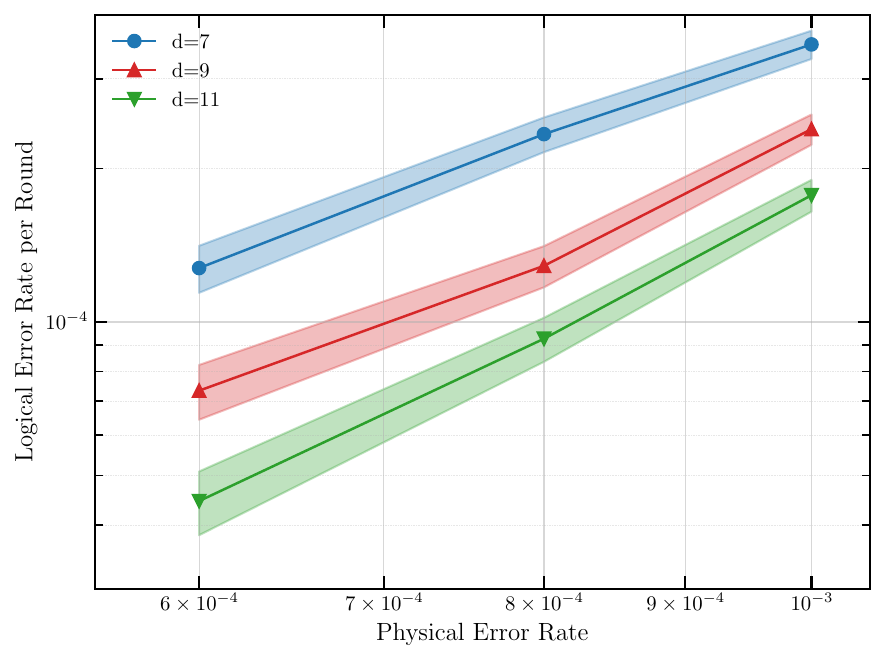}
	\caption{\textbf{Simulation results of superstabilizer handling of defect clusters: }coupler defects are not considered, and data qubits and ancilla qubits each have a 1\% probability of being defective.}
	\label{fig:defectcluster_sim}
\end{figure}
After performing the steps illustrated in Fig.~\ref{fig:defectcluster_diag}, we obtain a well-defined set of superstabilizers that effectively compensate for the defect cluster. Subsequently, it suffices to have all functioning ancilla qubits measure the gauge operators formed by their neighboring functioning data qubits in the quantum circuit, yielding superstabilizer syndromes for subsequent decoding. A detailed discussion on handling the boundaries of defect clusters is provided in Appendix~\ref{app:Boundary}.The pseudocode for superstabilizer handling of defect clusters is as follows:
\begin{algorithm}[H]
	\caption{Superstabilizer Generation for Color Code Defects}
	\label{alg:short}
	\begin{algorithmic}
		\State \textbf{Input:} stabilizer set $S$, data defects $D$, ancilla defects $A$, coupler defects $C$
		\State \textbf{Output:} updated $S$
		
		\State \textbf{Step 0:} For each $a \in A$, add its adjacent ancilla(s) to $A$  \Comment{per corrected description}
		
		\State \textbf{Step 1:} For each $d \in D$, call \textsc{HandleDataDefect}$(d)$  \Comment{remove one weight‑1 stabilizer on $d$, multiply others by it}
		
		\State \textbf{Step 2:} Process corner ancilla and coupler defects
		\State $R \gets \emptyset$  \Comment{defects already reduced}
		\While{changes occur}
		\State $T \gets (A \cup \text{ancillas of }C) \setminus R$
		\For{each $a \in T$}
		\State Let $X_a$, $Z_a$ be the unique X‑ and Z‑stabilizer involving $a$
		\If{$\exists$ data qubit $q$ in support$(X_a)\cap$support$(Z_a)$ with exactly one X‑ and one Z‑stabilizer}
		\State Remove the corresponding stabilizer(s) from $S$
		\State $D \gets D \cup \{q\}$, $R \gets R \cup \{a,\text{related defects}\}$
		\State \textbf{break} and restart the while loop
		\EndIf
		\EndFor
		\EndWhile
		
		\State \textbf{Step 3:} Reduce remaining ancilla and coupler defects to data qubits
		\For{each $c \in C \setminus R$}
		\If{both ends of $c$ are ancillas} $A \gets A \cup \{\text{its ancillas}\}$
		\Else $D \gets D \cup \{\text{data qubits of }c\}$\EndIf
		\EndFor
		\For{each $a \in A \setminus R$} $D \gets D \cup \{\text{neighbors of }a\}$ \EndFor
		
		\State \textbf{Step 4:} For each $d \in D$, call \textsc{HandleDataDefect}$(d)$
		\State \Return $S$
	\end{algorithmic}
\end{algorithm}

Beyond determining the superstabilizers, finding a logical operator that commutes with all gauge operators is essential for logical measurement and lattice surgery in color codes. When a stabilizer is discarded, if its support set overlaps with the logical operator, the logical operator must be updated by multiplying it with the discarded stabilizer to ensure commutativity with all gauge operators.

It is worth noting that the iSWAP-Mediated scheme proposed for isolated ancilla qubit defects can be seamlessly integrated into the above adapter workflow: before the third step, all ancilla and coupler defects that meet the usage conditions (See Appendix~\ref{app:iSWAP-Mediated} for details.) are screened out and processed using the iSWAP-Mediated scheme. In contrast, the neighbor-assisted scheme, due to its specific characteristics, currently lacks a general method for integration.

Fig.~\ref{fig:defectcluster_sim} presents our numerical simulation results. For color codes with distances 7, 9, and 11, we randomly generated 1000 lattices with a defect rate of 1\% and performed 1000 simulations on each lattice using the superstabilizer scheme. The simulation results show that, under the same defect rate, a larger code distance yields a lower logical error rate, and as the physical error rate increases, the logical error rate increases correspondingly for all code distances. This trend demonstrates that the superstabilizer scheme effectively mitigates the impact of defects and that increasing the code distance further enhances fault tolerance, thereby validating the scalability of the scheme.
\section{Logical Operator}
\label{sec:Operator}
Universal fault-tolerant quantum computation requires Clifford gates along with at least one non-Clifford gate \cite{campbell2017roads}. One of the key advantages of color codes is their ability to implement Clifford gates transversally \cite{moussa2016transversal,kubica2015universal}. In contrast, non-Clifford gates are typically realized in discussions of universal fault-tolerant quantum computation based on color codes through techniques such as lattice surgery combined with magic state preparation (e.g., magic state distillation \cite{lee2025low,landahl2014quantum} or magic state cultivation \cite{gidney2024magic}). In this section, we describe transversal Clifford gates on defective color codes as well as lattice surgery operations.
\subsection{Transversal Clifford Gate}
\label{subsec:Clifford}
In this section, we discuss the implementation of transversal Clifford gates on defective color codes.  Fig.~\ref{fig:Clifford}(a)(b)(c) illustrates the transversal implementation of the H, S and CNOT gates on defect-free color codes \cite{kubica2015universal,landahl2011fault,thomsen2024low}.

\begin{figure*}[htbp]
	\centering
	\includegraphics[width=0.8\linewidth]{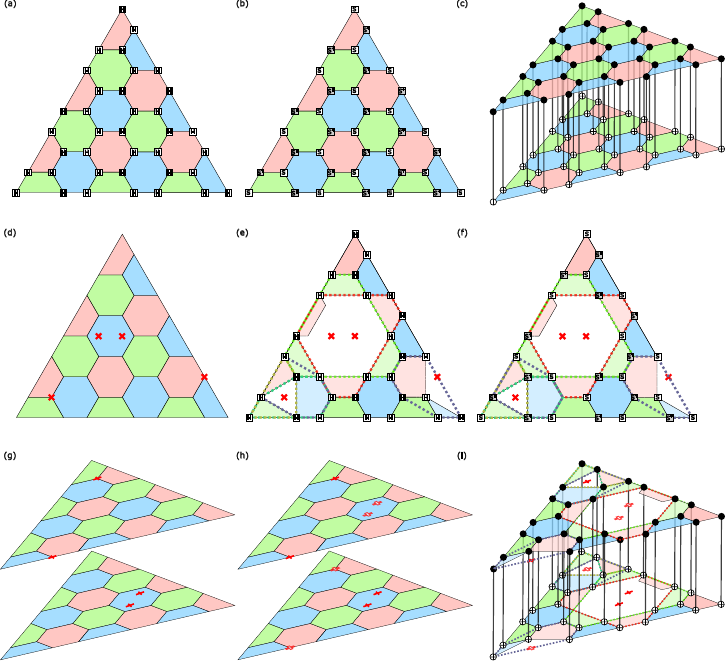}
	\caption{\textbf{Transversal Clifford gates on defect-free color codes:} (a) Transversal H gate; (b) Transversal S gate; (c) Transversal CNOT gate.	\textbf{Single-qubit Clifford gates on defective color codes:} (d) Schematic of defects; (e) Transversal H gate; (f) Transversal S gate.	\textbf{Transversal CNOT gate on defective color codes:} (g) Schematic of defects; (h) Because the defect distributions differ between the two code blocks, the stabilizers are not guaranteed to remain invariant under the transversal CNOT; therefore, additional defects are artificially introduced to make the defect distributions identical; (i) Transversal CNOT gate.}
	\label{fig:Clifford}
\end{figure*}

For the transversal H gate, given that the two-dimensional color code is a self-dual CSS stabilizer code, its logical Pauli operators are \(\bar{X}=X(Q)\) and \(\bar{Z}=Z(Q)\), where \(Q\) denotes the set of all data qubits. The self-dual property ensures that applying the Hadamard gate transversally on all data qubits, denoted \(H(Q)\), leaves the stabilizers unchanged. Furthermore, it guarantees the conjugation relations \(\bar{X} \rightarrow H(Q) X(Q) H(Q)^\dagger = \bar{Z}\) and \(\bar{Z} \rightarrow H(Q) Z(Q) H(Q)^\dagger = \bar{X}\). Therefore, \(H(Q)\) directly implements the logical H gate transversally.

For the transversal S gate, the data qubits of the color code form a bipartite graph, and each stabilizer has an even weight. As shown in Fig.~\ref{fig:Clifford}(b), one can apply S gates to a subset of data qubits and \(S^\dagger\) gates to the complementary set. Since the S gate represents a rotation around the Z-axis, all Z-type stabilizers remain invariant. For X-type stabilizers, utilizing the relations \(S^\dagger X S = iXZ\) and \(S X S^\dagger = -iXZ\), and noting that each stabilizer contains an equal number of S and \(S^\dagger\) gates, the coefficients cancel out. Consequently, the X stabilizer transforms into an XZ-type stabilizer, which, when multiplied by the corresponding Z stabilizer, returns to its original form. Moreover, because the number of S gates applied to data qubits exceeds the number of \(S^\dagger\) gates by one, the logical operator \(\bar{X}\) transforms as \(\bar{X} \rightarrow i\bar{X}\bar{Z} = \bar{Y}\), thereby realizing the transversal S gate.

In defective color codes, the implementation of these two Clifford gates remains straightforward. Whether employing the general superstabilizer approach, the Neighbor-Assisted scheme, or the iSWAP-Mediated scheme (which optimize ancilla qubits for defects), the stabilizers preserve the self-dual property. After introducing defects, the logical Pauli operators are obtained by multiplying the original logical operators with the defective stabilizers. Due to the self-duality of the stabilizers, the logical X and Z operators after defect introduction maintain identical supports. Hence, applying H gates transversally to all non-defective data qubits suffices to implement the logical H gate. Similarly, the self-dual nature of stabilizers also implies that the intersection of the supports of any two stabilizers has even cardinality (otherwise, they would not commute). Therefore, the transversal S gate implementation used for defect-free codes can be adopted, with the instruction that gates are not applied to defective qubits. The superstabilizers, formed by multiplying stabilizers, ensure that the difference between the count of S and \(S^\dagger\) gates modulo 4 remains zero. Likewise, after incorporating the defective stabilizers, the logical operator still satisfies the transformation \(\bar{X} \rightarrow \bar{Y}\).

For chips capable of implementing transversal CNOT gates (which typically implies long-range coupling capabilities), defect issues are often not a primary concern. However, to provide a comprehensive discussion of Clifford gate implementation on defective chips, we address this case here. A logical CNOT gate can be implemented transversally by applying physical CNOT gates to each pair of corresponding qubits across two identical copies of the color code. In the defect-free scenario, \(\overline{CNOT}\) satisfies: \(\overline{CNOT}(S_1^X \otimes I) = S_1^X \otimes S_2^X\) and \(\overline{CNOT}(I \otimes S_2^Z) = S_1^Z \otimes S_2^Z\). Multiplying these results by \(S_2^X\) and \(S_1^Z\), respectively, leaves the stabilizers unchanged, and the logical Pauli operators also follow the transformation rules of a logical CNOT. However, in defective color codes, the defect distributions in the two code blocks are often different, leading to inconsistent constructions of their respective superstabilizers. Consequently, applying the transversal CNOT gate no longer guarantees that the stabilizers remain unchanged. To implement a transversal CNOT under these conditions, one must simultaneously consider the defect distributions of both code blocks to align their superstabilizer constructions. This undoubtedly results in higher logical error rates and increased resource overhead.
\subsection{Lattice Surgery}
\label{subsec:Surgery}
\begin{figure}[htbp]
	\centering
	\includegraphics[width=\linewidth]{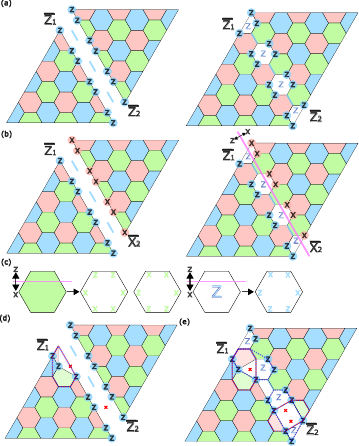}
	\caption{\textbf{ Schematic diagram of lattice surgery in color codes.}(a) Measurement of the logical operator $\bar{Z}_1\bar{Z}_2$ in a defect-free color code. Left panel: The intermediate blue auxiliary region is initialized into Bell states. The data qubits marked with Z on the boundaries correspond to the logical operators $\bar{Z}_1$ and $\bar{Z}_2$, respectively. Right panel: After merging, the central blue region contains only Z-type stabilizers.
		(b)Measurement of the logical operator $\bar{Z}_1\bar{X}_2$ in a defect-free color code. A domain wall is introduced, and anyonic excitations crossing it are converted between X-type and Z-type via a twist, enabling the measurement of non-commuting logical operators.
		(c) Schematic of a twist structure.
		(d) Initialization step for lattice surgery with defective color codes. A data qubit defect is present in both code patch 1 and the intermediate auxiliary region. Blue and red icons represent Z-type and X-type hyper-stabilizers, respectively.
		(e) Measurement of the logical operator $\bar{Z}_1\bar{Z}_2$ in a defective color code. The product of all blue stabilizers (including hyper-stabilizers that incorporate blue gauge operators) connecting the two logical operators yields the measurement value of $\bar{Z}_1\bar{Z}_2$.}
	\label{fig:lattice_surgery}
\end{figure}

Lattice Surgery is a key methodology in fault-tolerant quantum computing for implementing quantum logic gate operations, particularly multi-qubit entangling gates \cite{chatterjee2025lattice,tan2024sat}. Its core principle lies not in directly manipulating the protected logical qubits, but in achieving logical evolution by altering the structure of the underlying topological quantum error-correcting code, such as the surface code \cite{litinski2018lattice,litinski2019game} or the color code \cite{lee2025low,wang2025magic}. This technique requires only local interactions on a two-dimensional lattice and features relatively low resource overhead, making it one of the most promising architectures for scalable, universal quantum computation.

Implementing lattice surgery in color codes typically involves three steps: initialization, merging, and splitting. As shown in Fig.~\ref{fig:lattice_surgery}(a), when two color code patches are adjacent along a blue boundary, the intermediate blue auxiliary region must first be initialized. This involves preparing a Bell state on each blue edge within this region. Subsequently, specific stabilizer operators spanning the two original patches and the auxiliary region are measured (the merge operation), fusing the two logical blocks into one. For instance, this operation in Fig.~\ref{fig:lattice_surgery}(a) allows for the indirect measurement of the logical operator $\bar{Z_1}\bar{Z_2}$. Similarly, using the twist structure shown in Fig.~\ref{fig:lattice_surgery}(c), a measurement of the logical operator $\bar{Z_1}\bar{X_2}$ can be implemented as in Fig.~\ref{fig:lattice_surgery}(b). Finally, performing Bell-basis measurements on the qubits in the auxiliary region and applying appropriate corrections (the split operation) separates the merged patch back into individual code blocks.

For the surface code, the merged structure can be viewed as a larger surface code (potentially containing defects); therefore, most existing defect-handling schemes can be directly applied during lattice surgery operations for the surface code \cite{leroux2025snakes}. For defective color codes, although the post-merge structure is not identical to a standard color code, it still falls within the stabilizer code formalism. Consequently, our proposed superstabilizer adaptive scheme remains applicable, with its handling procedure illustrated in Figs.~\ref{fig:lattice_surgery}(d) and (e). Notably, our treatment of data-qubit defects in the auxiliary region is as follows: we choose to disable both the defective qubit and its Bell-pair partner. This decision is based on two main considerations: first, initializing a single data qubit in the intermediate region to meet the requirements of lattice surgery is challenging; second, based on prior experience with surface code research \cite{lin2024codesign}, single-qubit gauge operators are generally detrimental to the logical error rate. For lattice surgery involving a domain wall, as shown in Fig.~\ref{fig:lattice_surgery}(b), the post-defect handling is similar and is not elaborated here.

Beyond the superstabilizer scheme, our other proposed optimization strategies for handling ancilla qubit defects can also be applied in the lattice surgery context. In the junction region, the iSWAP-Mediated scheme handles intermediate defective regions in a similar manner: although additional measurements of two-body stabilizers in the intermediate region are required, the time steps saved by measuring single-Pauli-type stabilizer regions can compensate for the measurements of the two-body stabilizers. The Neighbor-Assisted scheme follows a similar approach. Additionally, it is worth noting that for ancilla qubit defects located in a single-Pauli-type stabilizer region—where simultaneous consideration of both \(X\)-type and \(Z\)-type stabilizers is no longer needed—the Neighbor-Assisted scheme can fully maintain the code distance for that Pauli type while reducing the code distance for the other Pauli basis by 2.

This section does not provide explicit circuits for performing lattice surgery but merely outlines feasible schemes. More detailed circuit optimization strategies specific to lattice surgery scenarios remain to be investigated in future work.
\section{CONCLUSION AND OUTLOOK}
\label{sec:CONCLUSION}
This paper establishes a universal framework for constructing superstabilizers applicable to arbitrary stabilizer codes. The proposed recursive generation algorithm fully covers defective regions while ensuring strict commutation relations with gauge operators, thereby laying a theoretical foundation for a general defect-adaptive scheme that can be extended to other topological codes. Building on this framework, we further propose a fully automated adaptive adapter specifically for the color code on square lattices following the superstabilizer scheme. This adapter uniformly reduces ancilla qubit and coupler defects to data-qubit defects, enabling adaptive handling under defective scenarios. For isolated ancilla qubit defects, we present two optimization schemes: the Neighbor-Assisted scheme improves qubit utilization by reusing surrounding ancilla qubits and effectively reduces the logical error rate in specific defect cluster scenarios, while the iSWAP-Mediated scheme achieves zero code distance loss using CXSWAP gates and can be seamlessly integrated into the processing flow of defect clusters. Finally, this work provides feasible schemes for transversal Clifford gates and lattice surgery on defective color codes, offering a viable pathway for the low-overhead automated deployment of color codes on defective hardware and delivering scalable technical support for building color-code-based universal fault-tolerant quantum computing systems.

Although this work has achieved preliminary results in handling color code defects, several key issues remain worthy of further exploration in subsequent research. For example, the comprehensive performance of the Neighbor-Assisted and iSWAP-Mediated schemes in complex defect cluster scenarios requires more systematic analysis and discussion to clarify their applicable scopes and optimal strategies. In addition, developing an efficient and scalable retrieval algorithm for locating the shortest logical error chains in defective color codes remains an important direction that warrants further investigation. Furthermore, the present work primarily focuses on validating the feasibility of the proposed schemes. Future research can build upon this foundation to explore performance optimization, such as attempting more powerful decoding algorithms or integrating strategies like stabilizer weight reduction and the shell approach \cite{strikis2023quantum} with the current schemes to enhance overall error correction performance.

Based on the propositions we have proved regarding the universality of the superstabilizer scheme, the superstabilizer framework proposed in this work can be extended to other error-correcting codes realizable on superconducting hardware platforms, such as the heavy-hexagon code \cite{chamberland2020topological,kim2026magic} or qLDPC-like BB codes \cite{wang2026demonstration}, to verify its cross-platform generality. At the same time, small-scale experiments on defective color codes can be carried out on actual superconducting chips to further validate and refine the proposed schemes through hardware testing, thereby advancing their practical implementation.

\section{Acknowledgments}
\label{sec:Acknowledge}

This work was supported by the National Key Research and Development Program of China (Grant No. 2024YFB4504101) and Anhui Province Innovation Plan for Science and Technology (Grant No. 202423r06050002). 

\appendix

\section{Neighbor-Assisted Scheme}
\label{app:Neighbor-Assisted}
The Neighbor-Assisted scheme is analogous to the SNL scheme in surface codes \cite{leroux2025snakes}, as both address defects by reusing neighboring ancilla qubits. However, due to the higher density of stabilizers in color codes, it is not possible to devise a scheme, akin to SNL, that preserves the code distance in one Pauli basis while reducing it by only 2 in the other basis when ancilla qubit defects are present. Specifically, if one attempts to preserve the code distance in the \( Z \) basis, the Neighbor-Assisted approach must be applied to the \( X \)-type stabilizers (as shown in the left panel of Fig.~\ref{fig:Neighbor_Assisted}(a)), thereby retaining all \( X \)-type stabilizers. This, however, necessarily results in the loss of the defective \( Z \)-type stabilizer; otherwise, the non-commutation between the \( X \)-type stabilizers and the \( Z \)-type gauge operators would force their combination into a superstabilizer. In this case, the missing \( Z \)-type stabilizer can only be compensated via the superstabilizer scheme (as illustrated in the right panel of Fig.~\ref{fig:Neighbor_Assisted}(a)) through multiplication. However, the resulting gauge operators effectively act as shorter error chains for the \( X \)-type stabilizers, causing the code distance in the \( X \) basis to decrease. Consequently, it is not possible to sacrifice the code distance in only one basis as in the SNL scheme. Instead, the Neighbor-Assisted scheme must be applied to both bases uniformly, as shown in the left panel of Fig.~\ref{fig:Neighbor_Assisted}(a). In this approach, blue ancilla qubits are used to measure the defective stabilizer, while the three surrounding red stabilizers multiply to form a superstabilizer. After this multiplication, the syndromes of these three red stabilizers become equivalent, meaning that within the bosonic theory, a red boson can move freely among them, forming shortcuts for error propagation.
\begin{figure}[htbp]
	\centering
	\includegraphics[width=0.8\linewidth]{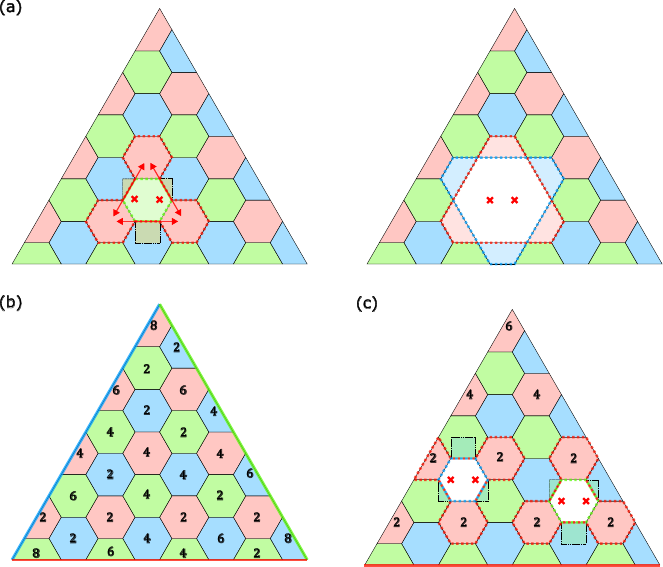}
	\caption{(a) Left: Schematic of the Neighbor-Assisted scheme for handling ancilla qubit defects; Right: The superstabilizer scheme. In the left panel, blue ancilla qubits are used to measure the defective stabilizer, and the three surrounding red stabilizers multiply to form a superstabilizer, allowing a red boson to move freely among them. (b) Labeling of stabilizers in a defect-free color code, where the number indicates the minimum weight of the error chain required to transport the corresponding boson to its same-color boundary. (c) An example where two ancilla qubit defects are both addressed using the Neighbor-Assisted scheme. The six resulting red boson shortcuts are all located between red stabilizers labeled 2 and 4, resulting in a maximum reduction of 2 in the minimum error chain weight required for red boson transport to the boundary.}
	\label{fig:Neighbor_Assisted}
\end{figure}
According to the bosonic theory of color codes \cite{kesselring2024anyon}, any logical operator can be viewed as transporting red, green, and blue bosons from the bulk to their respective color boundaries. In Fig.~\ref{fig:Neighbor_Assisted}(a), the non-commutation induced by the gauge operators measuring the defective stabilizer forces the three neighboring red stabilizers to merge into a superstabilizer, thereby establishing three shortcuts for red bosons among these stabilizers. To quantify the impact of such shortcuts on the code distance, each stabilizer in a defect-free color code can be assigned a numerical value representing the minimum weight of the error chain required to transport a boson from that stabilizer to its corresponding color boundary (as shown in Fig.~\ref{fig:Neighbor_Assisted}(b)). Based on this labeling, Fig.~\ref{fig:Neighbor_Assisted}(c) presents a representative example where two ancilla qubit defects are both addressed using the Neighbor-Assisted scheme. The six resulting red boson shortcuts are all located between red stabilizers labeled 2 and 4. Consequently, the minimum weight of the error chain required to transport a red boson to the red boundary is reduced by at most 2. In contrast, applying the superstabilizer scheme to the same configuration would lead to a code distance reduction of 4.

Although the Neighbor-Assisted scheme can significantly reduce the logical error rate in specific configurations (such as the one shown in Fig.~\ref{fig:Neighbor_Assisted}(c)), establishing a universal theory to rapidly determine the exact circuitry of this scheme and its precise impact on the code distance for more complex defect distributions remains challenging. Moreover, the temporal halving of the code distance implies that this scheme incurs higher resource overhead. Therefore, the advantages of the Neighbor-Assisted scheme are primarily manifested in certain specific circuit configurations.

\section{iSWAP-Mediated scheme}
\label{app:iSWAP-Mediated}

\begin{figure}[htbp]
	\centering
	\includegraphics[width=0.8\linewidth]{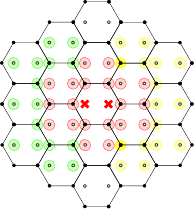}
	\caption{\textbf{Usage conditions of the iSWAP-mediated scheme:} A defect occurring on any bit marked in red in the figure indicates that the iSWAP-mediated scheme cannot be applied. Additionally, the scheme is inapplicable when defects occur simultaneously on bits marked in both yellow and green.}
	\label{fig:avai_iM}
\end{figure}

\begin{figure*}[htbp]
	\centering
	\includegraphics[width=0.85\linewidth]{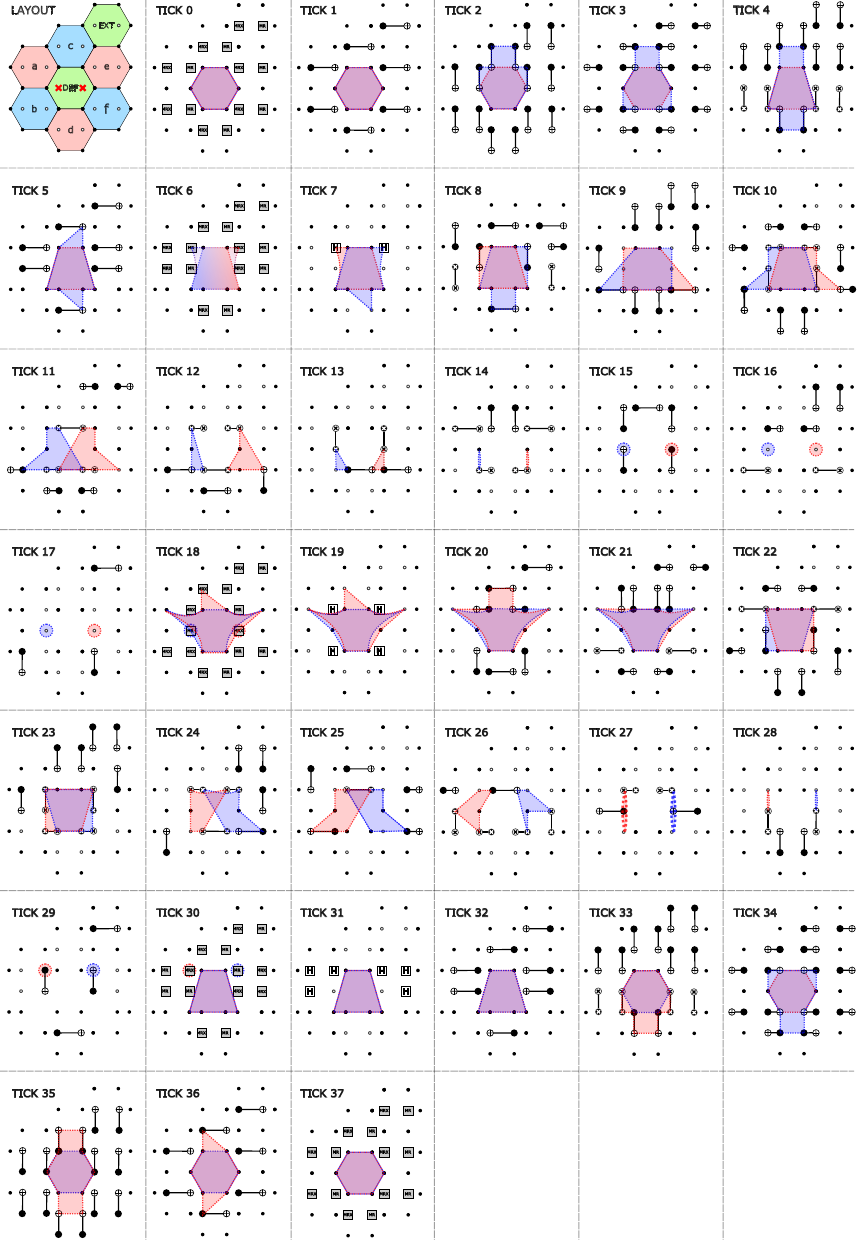}
	\caption{\textbf{ Complete cycle circuit of the iSWAP-mediated scheme.} The first subfigure shows the initial lattice and defect distribution, where six neighboring regions around the defect cluster are labeled as a through f, with DEF representing the defect cluster and EXT representing operations in all non-defect neighboring regions. During time steps 0–6, the circuit extracts syndromes for all non-defect X stabilizers. During time steps 7–18, it extracts syndromes for the defective X and Z stabilizers, as well as those of all non-neighboring X stabilizers, the X stabilizers in regions a–d, and the Z stabilizers in regions e and f. During time steps 19–30, it extracts syndromes for the defective X and Z stabilizers, along with those of all non-neighboring Z stabilizers, the Z stabilizers in regions a–d, and the X stabilizers in regions e and f. During time steps 31–37, the circuit extracts syndromes for all non-defect Z stabilizers and resets all qubits.In the figure, red/blue markings indicate X/Z stabilizers.}
	\label{fig:all_circuit_iM}
\end{figure*}
The iSWAP gate is equivalent to the CXSWAP gate, which can be understood as a SWAP gate following a CX gate. Applying a CXSWAP gate from an ancilla qubit to a data qubit is equivalent to extracting the X syndrome of the ancilla qubit followed by a swap operation. Leveraging this property, the iSWAP-mediated scheme enables ancilla qubits to continuously extract data qubit syndromes and perform swaps within an acceptable time overhead using CXSWAP gates. In this section, we present the complete circuit of the iSWAP-Mediated scheme.

\begin{figure*}[htbp]
	\centering
	\includegraphics[width=0.85\linewidth]{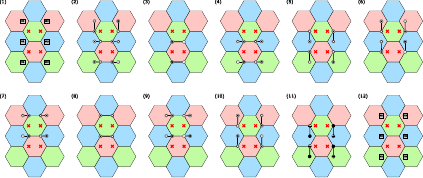}
	\caption{\textbf{Schematic illustration of the iSWAP-Mediated scheme for handling adjacent auxiliary-qubit defect clusters.} In the figure, syndromes of two pairs of X/Z defective stabilizers are extracted within 10 time steps, and all data qubits remain in place after one cycle without requiring resetting. Since syndromes of adjacent ancilla qubits are no longer extracted, the complete circuit is divided into two cycles: the first cycle measures all non-defective X and Z stabilizers; the second cycle measures the defective stabilizers, as well as all non-defective X and Z stabilizers that are not adjacent to the defective stabilizers.}
	\label{fig:doubleancilla_iM}
\end{figure*}

Compared with similar schemes in surface codes \cite{zhou2024halma}, the iSWAP-mediated scheme not only requires measuring a larger number of data qubits (with the stabilizer weight increasing from four to six) but also enjoys a more generous timing budget, as the defect-free circuit cycle extends from 4 to 10 two-qubit gates. This affords the iSWAP-mediated scheme greater flexibility and more possibilities. We observe that if we were to measure only the defective stabilizers within the 10 two-qubit gates, as done in \cite{zhou2024halma}, there remains significant redundancy. Therefore, we opt to simultaneously extract the X/Z syndromes of the neighboring regions while extracting the defective stabilizer syndromes, as illustrated in Fig.~\ref{fig:all_circuit_iM}. Within a full cycle, we employ 30 two-qubit gates and perform two rounds of syndrome extraction for all stabilizers. As a result, the required time overhead is reduced to 1.5 times that of the defect-free case, compared with twice that (if only defective stabilizers were measured, an additional round of neighboring stabilizer syndromes would be needed, requiring extraction of four rounds of syndromes for all non-neighboring defective stabilizers and two rounds for defective stabilizers and their neighboring ones within 40 two-qubit gates).

Regarding the handling of defect clusters, the applicable conditions of the iSWAP-Mediated scheme are illustrated in Fig.~\ref{fig:avai_iM}. When there exist isolated auxiliary-qubit defect pairs satisfying these conditions, the iSWAP-Mediated scheme can be integrated into our adaptive color code adapter. It is noted that in the figure, the qubits marked in yellow and green are divided into two groups and arranged symmetrically. This is because, when extracting syndromes for defective stabilizers, the Pauli types of neighboring defective stabilizers are inconsistent—due to constraints on the number of two-qubit gates and the required order for measuring adjacent stabilizers of different types, it is not possible to measure neighboring defective stabilizers of the same Pauli type within a single round. Consequently, the Pauli types of the two neighboring defective stabilizers on one side of the defect region differ from those of all other non-defective stabilizers. If additional defects exist near the boundary between stabilizers of different Pauli types, causing gauge operators to appear in the vicinity of the defect, these gauge operators may not commute with each other, leading to the failure of the fault-tolerant scheme. It is worth noting that iSWAP-gate-based schemes are flexible, and we cannot guarantee that the circuits we provide are optimal. The discussion of defect-cluster applicability in this paper is based solely on our proposed scheme.

Furthermore, one way to enhance the flexibility of the scheme to better accommodate defect clusters is to discard the measurement of neighboring stabilizers when measuring defective stabilizers, which is consistent with the halma approach used in surface codes. Here we present an example (as shown in Fig.~\ref{fig:doubleancilla_iM}): at the cost of omitting measurements of neighboring stabilizers, we successfully handle two adjacent pairs of auxiliary-qubit defects using the iSWAP-Mediated scheme without any loss of code distance. However, discarding measurements of neighboring stabilizers necessitates an additional round of syndrome extraction for those neighbors, thereby increasing the time overhead from a factor of 1.5 to a factor of 2.

In practice, since the time overhead of measurement and reset operations on actual quantum chips is much higher than that of two-qubit gates, the actual time overhead added by our iSWAP-Mediated scheme is far below the theoretical upper bound of 1.5×. In practical applications, one may also consider increasing the number of time steps to improve flexibility for more complex defect configurations, which we will not elaborate further here.

\section{Boundary}
\label{app:Boundary}
Readers familiar with defect handling in surface codes may note that, unlike the approaches discussed for boundary issues in surface code defect scenarios \cite{wei2025low}, this work does not provide a separate discussion for boundaries in the context of color code defects. Instead, we uniformly adopt the superstabilizer scheme. This section details the handling methods for defects located on or near the boundaries of the color code.
\begin{figure}[htbp]
	\centering
	\includegraphics[width=0.8\linewidth]{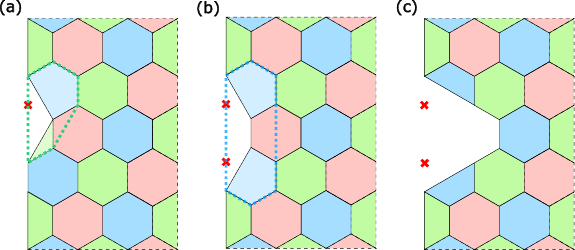}
	\caption{\textbf{Schematics of boundary defect handling:} (a) A single data qubit defect on the left boundary. The superstabilizer scheme multiplies two adjacent stabilizers to form a superstabilizer. (b) Two adjacent data qubit defects on the boundary. The three stabilizers associated with the defects are multiplied to form a superstabilizer. (c) The defect configurations in (a) and (b) share the same boundary deformation outcome. Due to the defects being located on the red boundary, a portion must be removed to ensure the new boundary is only contacted by blue-green stabilizers, resulting in a code distance reduction of 2.}
	\label{fig:Boundary_diagram}
\end{figure}
In conventional surface code schemes, defects on boundaries are typically addressed through "boundary deformation." This involves contracting the surface code patch according to specific rules until a new, defect-free patch with intact boundaries is formed. In the resulting patch, data qubits on the X/Z-type boundaries are stabilized by a combination of two X/Z-type and one Z/X-type stabilizers. Data qubits at the corners are stabilized individually by one X-type and one Z-type stabilizer. Extending this concept to the color code requires that data qubits on the red, green, or blue boundaries be stabilized by blue-green, red-green, or red-blue type stabilizers, respectively, while corner data qubits can only be stabilized by single-color stabilizers. However, due to the color periodicity inherent in the color code structure, applying such a boundary redefinition scheme to handle defects on color code boundaries often incurs a code distance reduction of 2, as illustrated in Fig.~\ref{fig:Boundary_diagram}(c).

If we instead continue to employ the superstabilizer approach, for a data qubit defect on a boundary (Fig.~\ref{fig:Boundary_diagram}(a)), we simply multiply the two X-type or Z-type stabilizers supported by that qubit, thereby forming two superstabilizers. For a data qubit defect at a corner (Fig.~\ref{fig:Boundary_diagram}(b)), we only need to remove the stabilizers associated with that corner. According to the proposition demonstrated in the main text, both of these methods result in a code distance reduction of at most 1.
\begin{figure}[htbp]
	\centering
	\includegraphics[width=0.8\linewidth]{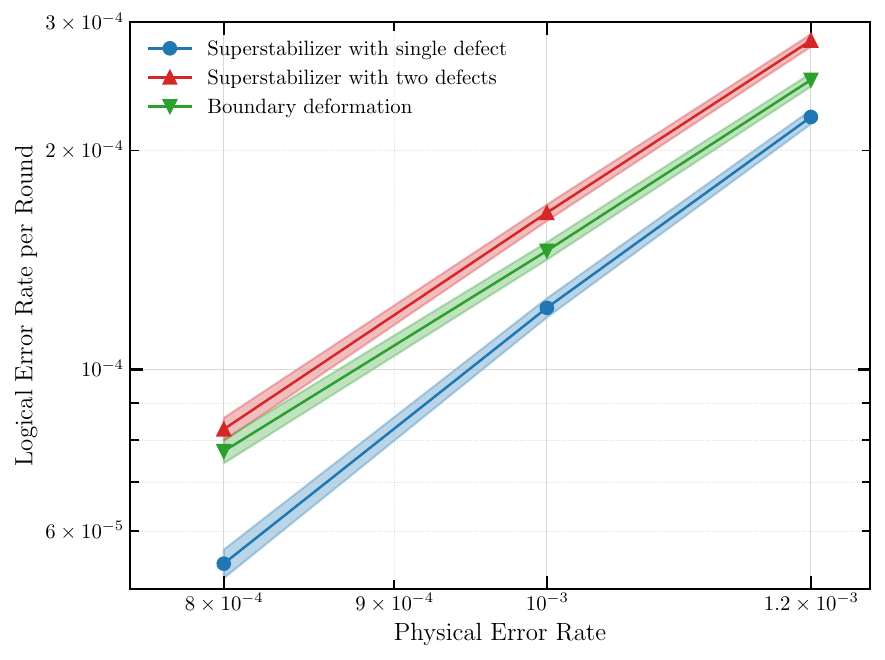}
	\caption{\textbf{ Simulations were performed on a color code with code distance \( d = 9 \).} The blue, red, and green curves correspond to the configurations shown in Fig.~\ref{fig:Boundary_diagram}(a), (b), and (c), respectively, with the exception of the pristine case where no defects are present on the boundary.}
	\label{fig:Boundary_sim}
\end{figure}
Fig.~\ref{fig:Boundary_sim} presents numerical simulation results for boundary defects. Consistent with our expectations, the superstabilizer scheme yields the lowest logical error rate when handling a single data qubit defect on the boundary. In scenarios requiring a code distance reduction of 2, the boundary redefinition scheme exhibits a lower logical error rate compared to the superstabilizer scheme applied to two adjacent boundary data qubit defects. This is attributed to the threshold degradation caused by the introduction of superstabilizers in the latter case. Fig.~\ref{fig:Corner} shows simulation results for corner defects, where the superstabilizer scheme consistently demonstrates a lower logical error rate than the boundary deformation method across all comparisons.
\begin{figure}[htbp]
	\centering
	\includegraphics[width=0.8\linewidth]{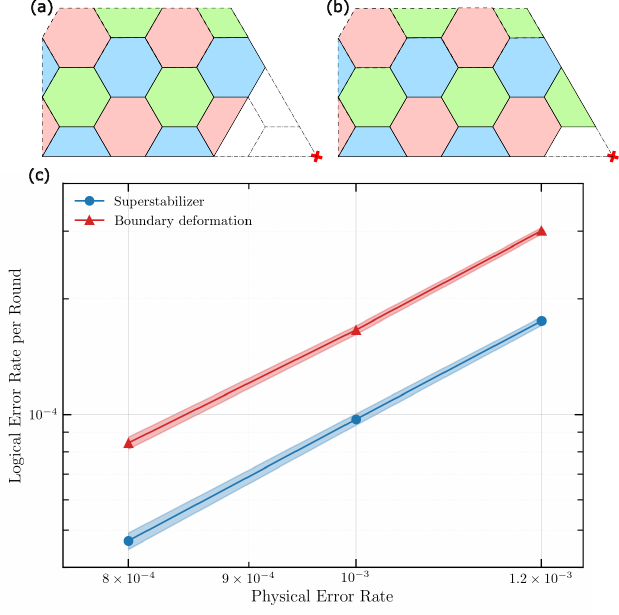}
	\caption{\textbf{Corner defect configurations.}(a) The boundary deformation scheme applied to a corner defect, resulting in a code distance reduction of 2. (b) The superstabilizer scheme applied to a corner defect, resulting in a code distance reduction of 1. (c) Numerical simulation results on a color code with code distance \( d = 9 \). The blue and red curves correspond to the configurations in (a) and (b), respectively, except for the pristine case without any defects on the boundary.}
	\label{fig:Corner}
\end{figure}
Considering that defect distributions are generally sparse in typical scenarios, this study ultimately adopts the superstabilizer scheme for handling defects on the boundaries of the color code.
\section{Details of numerical simulations}
\label{app:simulation}
All numerical simulations in this paper are conducted under the circuit-level noise model to obtain syndrome data. Specifically, the depolarizing Pauli noise channels are defined as follows:
\begin{equation}
	\begin{aligned}
		\mathcal{E}_{1}(\rho_1)&=(1-p)\rho_1+(p/3)\sum_{P\in\{X,Y,Z\}}P\rho_1P,\\
		\mathcal{E}_{2}(\rho_2)&=(1-p)\rho_2+(p/15)
		\times \\ &\sum_{\substack{P_1,P_2\in\{I,X,Y,Z\},\\P_1\otimes P_2\neq I \otimes I}}P_1\otimes P_2\rho_2P_1\otimes P_2,
	\end{aligned}
\end{equation}
where $\rho_1$ and $\rho_2$ are the single-qubit and two-qubit density matrices, respectively. In the simulated circuits, we apply $\mathcal{E}_{1}$ after single-qubit gates and $\mathcal{E}_{2}$ after two-qubit gates. Additionally, the measurement outcome and state initialization flip with probability $p$. We simulate the circuits using the Stim library~\cite{gidney2021stim}.

Subsequently, we employ the BP-OSD decoder for decoding~\cite{roffe2020decoding}. It should be noted that this decoding scheme is not optimal, which contributes to the relatively high logical error rates observed in our test results. In the future, we may consider exploring better decoding schemes for the defective color code, such as MWPM \cite{higgott2025sparse}(we speculate that this could be achieved by reducing the hyperedge stabilizer symptoms to defect-free symptoms and setting the weights of the edges corresponding to defective data qubits on the matching graph to zero), or by attempting recently promising machine learning decoding methods such as Alpha Qubit~\cite{bausch2024learning}.

For the tests with code distances 9 and 11 in Fig.~\ref{fig:sim_iM} of the main text, at least 100,000 trials are performed. For all other experiments, more than 1,000,000 samples are collected. The standard deviations for all data points are indicated by filled bands in the data figures.
\section{Stability Experiment}
\label{app:Stability}
\begin{figure}[htbp]
	\centering
	\includegraphics[width=0.85\linewidth]{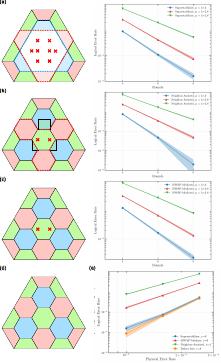}
	\caption{\textbf{Stability Experiments.} (a) superstabilizer scheme: The product of all complete red and green stabilizers with the red and green superstabilizers is the identity operator; "rounds" in the right subfigure denotes one full round of syndrome extraction for all stabilizers. (b) Neighbor-Assisted scheme: The defective stabilizer is measured using adjacent blue ancilla qubits; the product of the red stabilizers forms a superstabilizer, whose product with the green stabilizers yields the identity operator. (c) iSWAP-Mediated scheme: The defective stabilizer is measured via iSWAP gates using neighboring ancilla qubits. (d) Defect-free case: The product of all red and green stabilizers is the identity operator. (e) Comparison of stability experiment results for different schemes over the same physical duration. The defect-free case completes 8 measurement rounds; within the same duration, the Neighbor-Assisted scheme completes 4 rounds, and the iSWAP-Mediated scheme completes 6 rounds.}
	\label{fig:stability}
\end{figure}
In quantum computing, stability experiments refer to the design and execution of tests aimed at evaluating and verifying the long-term reliability and error suppression capabilities of qubits and the error-correcting codes they form, under noisy and decohering environments \cite{gidney2022stability,regent2025awesome}. Such experiments serve as a crucial benchmark for assessing the practicality of error-correcting codes, requiring the system to maintain the integrity of logical information throughout prolonged operation \cite{dasgupta2021stability,dasgupta2020characterizing}.

To this end, we designed the stability experiment framework shown in the left column of Fig.~\ref{fig:stability}. The defect-free case is illustrated in Fig.~\ref{fig:stability}(d), which includes 32 stabilizers to be measured and 30 data qubits. In the figure, the product of all same-type (X or Z) stabilizers marked in red and green equals the identity operator. To address defects in the ancilla qubit at the center of the code block, we adopted three different schemes: superstabilizer, iSWAP-Mediated, and Neighbor-Assisted, corresponding to the left subfigures of Fig.~\ref{fig:stability}(a), (b), and (c), respectively. The numerical simulation results for each scheme are presented in the right subfigures of Fig.~\ref{fig:stability}(a), (b), and (c). It can be observed that the error rates of all error-correcting schemes decrease exponentially with the number of measurement rounds.

From the comparison results in Fig.~\ref{fig:stability}(e), the superstabilizer scheme exhibits the best stability, nearly matching the performance of the defect-free case. This outcome differs significantly from typical results in surface codes, primarily because our superstabilizer scheme preserves the original code distance along the time direction without introducing additional time overhead. The iSWAP-Mediated scheme, which employs a large number of two-qubit gates during syndrome extraction, performs slightly worse. In contrast, the Neighbor-Assisted scheme, which involves higher-weight superstabilizers, shows the weakest performance.

A scheme that performs well in stability experiments suggests potential advantages when executing logical operations\cite{marton2025optimal,lin2024empirical}. However, the trade-off between this advantage and the spatial code distance remains a subject for further investigation.
\bibliographystyle{unsrtnat} 
\bibliography{references.bib}

\end{document}